\documentclass[twocolumn, superscriptaddress]{revtex4}
\usepackage{graphicx, amssymb, amsfonts, amsmath, amsthm, verbatim, changes, enumitem}
\newcommand\numberthis{\addtocounter{equation}{1}\tag{\theequation}}
\newcommand{\ket}[1]{\left| #1 \right\rangle}

\newcommand{\ketbra}[3][]{\mathinner{\lvert#2\rangle\! \langle #3\rvert}_{#1}}
\usepackage{bbm}
\newcommand*{\id}{{\normalfont\hbox{1\kern-0.15em \vrule width .8pt depth-.5pt}}}
\newtheoremstyle{mytheoremstyle} % name
    {\topsep}                    % Space above
    {\topsep}                    % Space below
    {}                   % Body font
    {}                           % Indent amount
    {\itshape}                   % Theorem head font
    {.}                          % Punctuation after theorem head
    {.5em}                       % Space after theorem head
    {}  % Theorem head spec (can be left empty, meaning ‘normal’)

\theoremstyle{mytheoremstyle}
\newtheorem{theorem}{Theorem}
\newtheorem{lemma}{Lemma}

\def\e{\mathrm{e}}
\def\i{\mathrm{i}}

\setlist[itemize]{leftmargin=*}
\begin{document}
\title{From single-shot to general work extraction with bounded fluctuations in work}
\author{Jonathan G. Richens}
\email{jonathan.richens08@ic.ac.uk}
\affiliation{Controlled Quantum Dynamics theory group, Department of Physics, Imperial College London, London SW7 2AZ, UK.}
\affiliation{Department of Physics and Astronomy, University College London,
Gower Street, London WC1E 6BT, UK.}
\author{Lluis Masanes}
\affiliation{Department of Physics and Astronomy, University College London,
Gower Street, London WC1E 6BT, UK.}

\begin{abstract}
\noindent In the standard framework of thermodynamics work is a random variable whose average is bounded by the change in free energy of the system. This average work is calculated without regard for the size of its fluctuations. We show that for some processes, such as reversible cooling, the fluctuations in work diverge. Realistic thermal machines may be unable to cope with arbitrarily large fluctuations. Hence, it is important to understand how thermodynamic efficiency rates are modified by bounding fluctuations. We quantify the work content and work of formation of arbitrary finite dimensional quantum states when the fluctuations in work are bounded by a given amount $c$. By varying $c$ we interpolate between the standard and min free energies. We derive fundamental trade-offs between the magnitude of work and its fluctuations. As one application of these results, we derive the corrected Carnot efficiency of a qubit heat engine with bounded fluctuations.  
\end{abstract}

\maketitle
\normalem
Historically, thermodynamics has been a theory of macroscopic systems comprising of many particles. As we venture away from the thermodynamic limit we must question the validity of established principles. Recently, the problem of extracting work from a microscopic quantum system has received much attention \cite{horodecki2013fundamental,skrzypczyk2014work,gemmer2015single,aaberg2013truly,dahlsten2011inadequacy}. The standard free energy is used to calculate the maximal amount of average work that can be extracted from a system in thermal contact with an infinite heat bath. Generally the work extracted on each running of the protocol fluctuates, but in the thermodynamic limit the relative size of fluctuations in work vanishes.  
However, in the case of microscopic systems, and systems that are far from equilibrium, fluctuations in the work can no longer be ignored. It is of significant practical importance that we understand these fluctuations in order to describe the behaviour of small and fragile machines such as quantum heat engines comprising of just a few qubits \cite{linden2010small,ryan2008spin,quan2006maxwell}. Realistic thermal machines are designed to operate at specific energies with a certain tolerance to fluctuations. Taking into account this inevitable fragility requires a modified free energy that tells us the average work associated with a process when fluctuations in that work are constrained.\\
\indent One approach to dealing with fluctuations is to simply not allow for them. This is the tactic employed by single-shot thermodynamics, a recently developed approach to quantum thermodynamics inspired by the field of single-shot information theory \cite{horodecki2013fundamental,aaberg2013truly,dahlsten2011inadequacy}. 
The single-shot (deterministic) work associated with a process is given by the difference in min-free energy between initial and final states \cite{horodecki2013fundamental,aaberg2013truly}, which is generally significantly smaller that the standard free energy difference. The work cost of forming a state from the Gibbs state is given by the max-free energy \cite{horodecki2013fundamental,aaberg2013truly}, which is generally significantly larger than the deterministic work that can be extracted from the state. This discrepancy between the work cost and work content of states in the single-shot regime results in thermodynamic irreversibility when transforming between states. Furthermore, the set of allowed thermodynamic transformations in the single-shot regime are severely restricted. In this regime, it is possible for a state to undergo a transition $\rho\rightarrow\rho'$ deterministically (and without supplying work) on the condition that an infinite family of ``second laws'' are satisfied \cite{brandao2015second}. Some transitions $\rho \leftrightarrow \rho'$ can only be achieved by supplying work in both forward and backward directions, resulting in a partial order on the set of states with respect to the resource of work \cite{brandao2015second}. This is in stark contrast to when we allow work to fluctuate freely, whereby all states can be inter-converted in a thermodynamically reversible manner. %One may be tempted to assume that the fundamental irreversibility and partial order are idiosyncratic of the single-shot regime, due to the strong constraint of requiring work to be a deterministic quantity. In this article we show that there are always state transformations that can only be achieved in a thermodynamically irreversible way, or exhibiting a partial order with respect to work, if we demand only that work cannot fluctuate infinitely. \\ 

\indent To date the majority of thermodynamic protocols treat work either as an unconstrained random variable or a totally constrained (deterministic) quantity. In this article we explore the landscape of protocols that exist between these two regimes. We find that in many protocols, for example thermodynamically reversible cooling, the work must have fluctuations that diverge in size. 
This makes realising these protocols practically infeasible, especially for small or fragile machines. To this end we define the $c$-bounded work, giving the optimal average work $\langle w\rangle$ that can be achieved by any protocol when fluctuations of the random variable $w$ are bounded as
\begin{eqnarray}
  \left|\, w-\langle w\rangle\, \right|
  \leq c \label{basic bound}
\end{eqnarray}
where $c$ is a adjustable parameter. In this article we explore how bounding work fluctuations in this way affects work extraction, state formation and the allowed state transformations of individual systems. We derive expressions for the $c$-bounded work that interpolate between these two regimes of deterministic and freely fluctuating work. We then apply these results to the study of a single qubit thermal engine, and derive a corrected Carnot efficiency when fluctuations in the work produced by the engine are constrained.\\ 

\section*{Results}

\noindent\textbf{The framework}. In this section we provide a precise description of our framework, describing the system, bath, work system and the set of allowed operations. \\
\indent We make use the widely applied set-up for thermodynamic protocols of system, infinite thermal bath and a weight, which acts as a store and source of the work produced or consumed by a protocol \cite{masanes2014derivation,skrzypczyk2014work,gemmer2015single}. 
In the following we set the Boltzmann constant $k_B$ to 1. The bath has infinite volume and it is in the Gibbs state $\rho_\text{B} = \frac 1 {\mathcal Z _\text{B}} e^{-\beta  H_\text{B}}$, where $\beta$ is the inverse temperature, $H_\text{B}$ the Hamiltonian, and $\mathcal Z _\text{B}$ the partition function.
%and a density of states given by $\Omega (E) = e^{S + \beta E}$ where $S$ is the entropy of the bath and $E$ is the bath energy. We take the energy levels of the bath to be close together and approximate the energy spectrum as continuous. In the canonical state the occupation probability of a bath state on energy level $E$ and degeneracy $g$ is given by $P(E,g) = \mathcal{Z}^{-1}_\text{B} e^{-\beta E}$ where $\mathcal{Z}_\text{B} =\text{tr}\left(e^{-\beta \mathcal{H}_\text{B}} \right)$ is the partition function of the bath where $\mathcal{H}_\text{B}$ is the bath Hamiltonian.  
\\
\indent The work system is modelled as a suspended weight with a continuous energy spectrum and Hamiltonian dependent only on its displacement $H_\mathrm{W}=  \int_\mathbb{R} dx\, x |x\rangle \! \langle x|$, where the orthonormal basis $\{|x\rangle, \forall\, x\in \mathbb R\}$ represents the position of the weight. 
In order to define work as a classical random variable $w$, the position of the weight is measured at the beginning and end of the protocol.
\\
\indent The system being transformed has Hilbert space of dimension $d$, initial state and Hamiltonian $(\rho, H_\text{S})$, and final state and Hamiltonian $(\rho', H'_\text{S})$ (which may have no relation to the initial Hamiltonian, see Supplementary Note 1).
It is useful to define the initial and final dephased states and their spectral decompositions

\begin{eqnarray}
  \lim\limits_{T\rightarrow\infty} \int\limits_0^T dt\, e^{-\i H_\text{S} t} \rho\, 
  e^{\i H_\text{S} t}
  &=&
  \sum_s x_s |s\rangle\! \langle s|
  \ ,
  \\
 \lim\limits_{T\rightarrow\infty} \int\limits_0^T dt\, e^{-\i H'_\text{S} t} \rho'\, 
  e^{\i H'_\text{S} t}
  &=&
  \sum_s x_{s'} |s'\rangle\! \langle s'|
  \ .
\end{eqnarray}

The two bases $|s\rangle$ and $|s'\rangle$ defined above, allow to write the spectral decompositions $H_\text{S} = \sum_s \mathcal E_s |s\rangle\! \langle s|$ and $H'_\text{S} = \sum_{s'} \mathcal E_{s'} |s' \rangle\! \langle s'|$.
(Note that we use notation $x_{s'}$ and $\mathcal E_{s'}$ instead of $x'_{s'}$ and $\mathcal E'_{s'}$.)
Finally, we assume that initially the joint state of system, bath and weight is product $\rho \otimes \rho_\text{B} \otimes \rho_\text{W}$. We consider any process that is a joint transformation of system, bath and weight represented by a Completely Positive Trace Preserving (CPTP) map $\Gamma_\mathrm{SBW}$ satisfying the following conditions:\\

\begin{comment}
\begin{description}
  \item[Microscopic reversibility (Second Law):] It has an (CPTP) inverse $\Gamma_\mathrm{SBW}^{-1}$, which implies unitarity $\Gamma_\mathrm{SBW} (\rho_\mathrm{SBW}) = U\rho_\mathrm{SBW} U^\dagger$.
    
  \item[Energy conservation (First Law):] 
$[U,H_\mathrm{S}\!+\!H_\mathrm{B}\!+\!H_\mathrm{W}]\!=\!0$.

  \item[Independence from the ``position" of the weight:]
  The unitary commutes with the translations on the weight $[U,\Delta_\mathrm{W}] = 0$. The generator of the translations $\Delta_\mathrm{W}$ is canonically conjugated to the position of the weight $[H_\mathrm{W}, \Delta_\mathrm{W}] = \i$.

  \item[Classicality of work:] Before and after applying the global map $\Gamma_\mathrm{SBW}$ the position of the weight is measured, obtaining outcomes $|x\rangle$ and $|x+w\rangle$ respectively. The joint transformation of the system and work random variable $w$ is given by the map
\begin{equation}
  \label{TO}
  \Lambda (\rho,w)
  =
  \int_\mathbb{R}\! dx\, {\rm tr}_\mathrm{BW}\! 
  \left[ Q_{x+w}\,
  U \left( \rho\otimes  \rho_\mathrm{B}\otimes 
  Q_x \rho_\mathrm{W} Q_x
  \right) U^\dagger \right]\ ,
\end{equation}
where $Q_x = |x\rangle\! \langle x|$ is a weight position projector.

%  \item[Independence from the state of the weight:] The map $\Gamma_\mathrm{S} (\rho_\mathrm{S},w)$ is independent of the initial state of the weight $\rho_\mathrm{W}$.
  
\end{description}
\end{comment}

\noindent Microscopic reversibility (Second Law): It has an (CPTP) inverse $\Gamma_\mathrm{SBW}^{-1}$, which implies unitarity $\Gamma_\mathrm{SBW} (\rho_\mathrm{SBW}) = U\rho_\mathrm{SBW} U^\dagger$.\\

\noindent Energy conservation (First Law): $[U,H_\mathrm{S}\!+\!H_\mathrm{B}\!+\!H_\mathrm{W}]\!=\!0$.\\

\noindent Independence from the ``position" of the weight: The unitary commutes with the translations on the weight $[U,\Delta_\mathrm{W}] = 0$. The generator of the translations $\Delta_\mathrm{W}$ is canonically conjugated to the position of the weight $[H_\mathrm{W}, \Delta_\mathrm{W}] = \i$.\\

\noindent Classicality of work: Before and after applying the global map $\Gamma_\mathrm{SBW}$ the position of the weight is measured, obtaining outcomes $|x\rangle$ and $|x+w\rangle$ respectively. The joint transformation of the system and work random variable $w$ is given by the map
\begin{equation}
  \label{TO}
  \Lambda (\rho,w)
  =
  \int_\mathbb{R}\! dx\, {\rm tr}_\mathrm{BW}\! 
  \left[ Q_{x+w}\,
  U \left( \rho\otimes  \rho_\mathrm{B}\otimes 
  Q_x \rho_\mathrm{W} Q_x
  \right) U^\dagger \right]\ ,
\end{equation}
where $Q_x = |x\rangle\! \langle x|$ is a weight position projector.\\

The assumption that the dynamics of a closed system is reversible and conserves energy is widely used, because it corresponds to a common physical setup. The third condition implies that the reduced map on the system and bath is a mixture of unitaries, and therefore cannot decrease the entropy of the joint state of system and bath (See Result 1 in \citep{masanes2014derivation}). This ensures that the weight cannot be used as a source of non-equilibrium, and can be viewed as a necessary condition for defining work \citep{skrzypczyk2014work, gallego2015defining}. 
As a consequence of the fourth condition (classicality of work), the optimal work that can be extracted is determined by the dephased states of the system, hence the presence of coherences in the system cannot increase or decrease this amount (See Supplementary Note 3). In the case of state formation, a state with coherences cannot be formed from a thermal state. However, for general state transformations (which we do not analyse in this paper) the presence of coherences in the initial and final states of system generate further constraints on the work \cite{lostaglio2015description,lostaglio2015quantum}. Regarding quantum definitions of work, several attempts have been made \cite{alhambra2016second,aaberg2014catalytic,perarnau2016quantum}, but there is still no consensus on these definitions and their treatment of fluctuations. The problem of defining a truly quantum definition of work, or if such a definition exists, remains an important and open question \cite{perarnau2016quantum,talkner2016aspects,gallego2015thermodynamic}. 
\\

\noindent \textbf{Deterministic work.} 
The single-shot work content of a system is given by the difference in min free energy $F^\text{min} (\rho)=-\beta^{-1}\log\sum_s x_s^0\, e^{-\beta \mathcal{E}_s}$ between the state $\rho$ and the thermal state 
%specified by $H_\text{S}$\\
\begin{equation}
W^{(0)}(\rho) = \frac{1}{\beta}\log\mathcal{Z}-\frac{1}{\beta}\log\sum\limits_s x_s^0\, e^{-\beta \mathcal{E}_s} 
\ ,
\label{minfreenergy}
\end{equation}
where $x^0$ returns 1 if $x> 0$ and 0 if $x=0$, and $\mathcal{Z} =\text{tr}\left(e^{-\beta H_\text{S}}\right)$ is the partition function of the system \citep{aaberg2013truly,horodecki2013fundamental}. 
If $x_s$ has full rank then min free energy is $-\beta^{-1}\log\mathcal{Z}$. Therefore non-zero deterministic work can only be extracted from states that are not of full rank. The single-shot work of formation is 
\begin{equation}
W^{(0)}_\text{F} (\rho ) = 
\frac 1 \beta \log \mathcal Z
+
\frac{1}{\beta}\log \max_s x_s \, e^{\beta \mathcal{E}_s} \ .
\end{equation}
In general we find that $W^{(0)}_F(\rho )>W^{(0)} (\rho)$, i.e. it is not possible to form most states in a thermodynamically reversible manner. 
When a weight is not present, the necessary and sufficient condition for a state transformation $(\rho , H_\text{S} ) \rightarrow (\rho' , H'_\text{S} )$ to be possible is given by the thermo-majorisation criteria~\cite{horodecki2013fundamental}. If in addition a catalyst is used, the necessary and sufficient conditions are given in \cite{brandao2015second} (for states that are diagonal in the energy eigenbasis). The key phenomenon is that in single-shot thermodynamics there is a partial order on states, 
%$(\rho , \mathcal{H}_s )  \nsucc\!\nprec(\sigma , \mathcal{H}'_s )$, 
i.e. there are state transitions that are impossible, in both forward and backward directions, without supplying work. In contrast thermodynamic irreversibility and partial order are not observed in the standard thermodynamic formalism, in which work is allowed to fluctuate freely, as we discuss next.\\

\noindent \textbf{Work with unbounded fluctuations}.  
The maximum average work that can be extracted from a system with the assistance of a heat bath with inverse temperature $\beta$ is given by the difference in free energy between $\rho$ and the Gibbs state. The free energy is given by $F(\rho) = \langle \mathcal{E} \rangle - \beta^{-1} S (\rho)$,  $\langle \mathcal{E} \rangle=\text{tr}\left(\rho \,H_\text{S}\right) = \sum_s x_s \mathcal E_s$ is the internal energy of the state, $S(\rho) =-\sum_s x_s \log x_s$ is the entropy of the de-phased state.
The Gibbs state, with energy level occupation probabilities $x_s = \mathcal{Z}^{-1} e^{-\beta \mathcal{E}_s}$, is the unique state given $H_\text{S}$ and $\beta$ with the lowest free energy (given by $-\beta^{-1} \log \mathcal{Z}$). Therefore the optimal average work that can be extracted from an out of equilibrium state is given by $W^{(\infty)}=\beta^{-1} \log \mathcal{Z}+ F(\rho)$. In the reverse process of state formation the work cost is also given by the difference in free energy between initial and final states. In other words, if we do not bound fluctuations in the work it is always possible to realise all state transformations in a thermodynamically reversible way. This poses the question - what is the minimum amount we must allow work to fluctuate in order for a transition to be achievable with a thermodynamically reversible protocol? 
\begin{theorem}\label{thermodynamic reversibility}
The thermodynamically reversible process achieving the transition $(\rho, \, H_\text{S})\rightarrow (\rho' , \, H_\text{S}')$ with minimal fluctuations in work has work values $w(s'|s) = \beta^{-1}\log (x_se^{\beta \epsilon_s}/x_{s'}e^{\beta \epsilon_{s'}})$. Therefore there exists a thermodynamically reversible process achieving the transition $(\rho, \, H_\text{S})\rightarrow (\rho' , \, H_\text{S}')$ with fluctuations in work less than or equal to $c$ if 
\begin{equation}
e^{\beta (\Delta F-c)}\leq \frac{x_s e^{\beta \mathcal{E}_s}}{x_{s'} e^{\beta \mathcal{E}_{s'}}}\leq  e^{\beta (\Delta F+c)} \quad \forall \, s, \, s'\ , \label{revcond}
\end{equation}
where 
%$x_s = \mathrm{tr}[\rho \ketbra{s}{s}]$, $x_{s'} = \mathrm{tr}[\rho' \ketbra{s'}{s'}]$, $\mathcal{H}_\text{S}\ket{s}=\mathcal{E}_s \ket{s}$ and $\mathcal{H}'_\text{S}\ket{s'}=\mathcal{E}_{s'} \ket{s'}$ and 
$\Delta F = F(\rho ) - F(\rho')$ 
is the change in the standard free energy. This becomes a necessary and sufficient condition if the initial and/or or final state is diagonal in the energy eigenbasis
\end{theorem}

\noindent \emph{Proof:} For a detailed proof see Supplementary note 2.\\

Note that for any finite $c$ there exist states such that \eqref{revcond} cannot be satisfied. These bounds have strong consequences for the minimal fluctuations that can be achieved with a thermodynamically reversible protocol. For example
\begin{equation}
\lim\limits_{x_i \rightarrow 0} \log \frac{x_s e^{\beta \mathcal{E}_s}}{x_{s'} e^{\beta \mathcal{E}_{s'}}}= - \lim\limits_{x_{s'} \rightarrow 0} \log \frac{x_s e^{\beta \mathcal{E}_s}}{x_{s'} e^{\beta \mathcal{E}_{s'}}} = - \infty
\end{equation}
Therefore as an energy level occupation probability tends to zero the work fluctuation associated with transitioning to or from this energy level diverges, negatively for work extraction and positively for state formation, when performing a thermodynamically reversible protocol. In either case we require $c\rightarrow\infty$ in order to satisfy inequalities \eqref{revcond}.\\
\indent Result \ref{thermodynamic reversibility} tells us that the further from equilibrium the initial or final states are, the larger the work fluctuations will be in
thermodynamically reversible protocols. Note that cooling a system close to its ground state is an example of transitioning from the thermal state to a far from equilibrium state. Similarly, if we want to extract work form a far from equilibrium state using a thermodynamically reversible transformation we encounter the same divergence in fluctuations. The fluctuations can diverge even if the average work remains small (for example, if the system is a qubit with trivial Hamiltonian then $W\leq \beta^{-1}\log 2$). These divergences have been previously noted in the recent study of absolute irreversibility \cite{murashita2015absolute, murashita2014nonequilibrium}.\\
%\indent For work extraction, the average work is given by $\beta^{-1}\sum_i x_i \log \left( x_i e^{\beta\mathcal{E}_i}\right)-\beta^{-1}\sum_j y_j \log \left( y_j e^{\beta\mathcal{E}_j}\right)=F(\rho) - F(\sigma) = W (\rho \rightarrow \sigma)$ so clearly in the case where and $x_i\rightarrow 0$ or a $y_j\rightarrow 0$ the fluctuations tend to infinity but but are attenuated by the probabilities tending to zero. States with low occupation probabilities that are far from equilibrium have
%\begin{equation}
%x_i \ll \frac{e^{-\beta \epsilon_i}}{\mathcal{Z}} \, \, , \text{ or} \, \, y_j \ll \frac{e^{-\beta \epsilon_j}}{\mathcal{Z}}
%\end{equation}
%The associated fluctuations can be made arbitrarily large $|w_{ij}- W(\rho \rightarrow \sigma )| \gg 1$.\\
\indent Previous discussions of the inadequacy of the standard free energy in the nano-regime have focused on the definitions of work \cite{dahlsten2011inadequacy, aaberg2013truly}. Here we add another criticism, that using the standard free energy to describe work we necessarily requires set-ups that can tolerate arbitrary fluctuations, which diverge in size for processes with initial or final states that are increasingly far from equilibrium.  \\

\noindent \textbf{Work with bounded fluctuations.} Motivated by these observations we define the c-bounded work content $W^{(c)} (\rho)$ as the maximum average work that can be extracted from state $\rho$ with initial Hamiltonian $H_\text{S}$ and final Hamiltonian $H'_\text{S}$, when the fluctuations of the work are constrained by $c$, as in \eqref{basic bound}. This notion of c-bounded work, and a generalisation to include a small probability of failure, was proposed in \citep{aaberg2013truly} but not developed beyond its definition.
Analogously, we define the c-bounded work of formation $W^{(c)}_\mathrm{F} (\rho)$ as the minimal average work that is necessary to create a state $\rho$ with  Hamiltonian $H'_\text{S}$ from the Gibbs state (with respect to initial Hamiltonian $H_\text{S}$), such that fluctuations in the work are bounded by $c$. 
\begin{theorem}
The $c$-bounded work content $W^{(c)}(\rho)$ and work of formation $W^{(c)}_F(\rho)$ are given by
\begin{eqnarray}
W^{(c)}(\rho)\! &=& \!\frac{1}{\beta}\log \mathcal{Z}'- \frac{1}{\beta}\log \sum_s x_s^0\, e^{-\beta (\mathcal{E}_s-\theta^{(c)}_{s})} \label{cwork}
\\
W^{(c)}_F(\rho)\! &=&\! \frac{1}{X_\mathrm{u}}\! \left[ \sum\limits_{s\in \mathcal X_\mathrm{u}}\!  \frac{x_s}{\beta} \log \!\left(x_s e^{\beta \mathcal{E}'_s}\!\mathcal{Z}\! \right)\!\!+\!c \, (1\!-\!X_\mathrm{u})\right] \label{cform}
\end{eqnarray}
\end{theorem}

\noindent \emph{Proof:} See Supplementary Notes 3 and 4 respectively.\\

First we describe the terms in \eqref{cwork}. $\mathcal{Z}'= \sum_{s'} e^{-\beta \mathcal{E} _{s'}}$ is the partition function of the final Hamiltonian $H'_\text{S}$. The second term can be viewed as a generalisation of the min free energy \eqref{minfreenergy} that allows for fluctuations in work $\theta^{(c)}_s$. The only difference to the min free energy is the term $e^{\beta \theta^{(c)}_{s}}$ included in the summation. $\theta^{(c)}_{s}$ is the fluctuation of the work value from the $c$-bounded average $W^{(c)}(\rho)$ given that the system was initially in state $\ket{s}$. In order to find the fluctuations associated with the optimal $c$-bounded work extraction protocol we must partition the energy levels into three disjoint subsets $\{1,2, \ldots ,d \} = \mathcal X_\mathrm{u} \cup \mathcal X_+ \cup \mathcal X_-$, representing the energy levels with positive $\mathcal X_+$, negative $\mathcal X_-$ and unbounded $\mathcal X_\mathrm{u}$ fluctuations. We also define
\begin{eqnarray}
  X_\mathrm{u} &=& \mbox{$\sum_{s\in \mathcal X_\mathrm{u}}$} x_s \ ,\\
  X_\pm &=& \mbox{$\sum_{s\in \mathcal X_\pm}$} x_s \ .
\end{eqnarray}

The algorithm for determining the partition $\mathcal X_\mathrm{u} \cup \mathcal X_+ \cup \mathcal X_-$, which requires the checking of at most $d-1$ inequalities, is described Supplementary Note 3B. Once we have determined the partition, the fluctuations are given by \\
\begin{equation}\label{c bounded flucs}
\theta^{(c)}_s = \begin{cases}
\frac{1}{\beta}\log (x_s e^{\beta \mathcal{E}_s})-\nu,\quad  s\in \mathcal X_\mathrm{u} \\
+c , \quad s\in \mathcal X_+ \\
-c ,\quad s\in \mathcal X_-
\end{cases}
\end{equation}
where
\begin{eqnarray}
  \nu &=& \frac{1}{X_\mathrm{u}}\left(F_\mathrm{u} + c(X_+ - X_-) \right)
  \\
  F_\mathrm{u} (\rho) &=& \sum_{s\in \mathcal X_\mathrm{u}} x_s \log (x_s e^{\beta \mathcal{E}_s})
\end{eqnarray}
where $F_\mathrm{u} (\rho)$ is the un-normalised free energy calculated for the unbounded partition only. Note that $W^{(c)}(\rho)$ can be written in the more compact form\\
\begin{equation}
W^{(c)}(\rho) = \frac{1}{\beta}\log\mathcal{Z}'-\frac{1}{\beta}\log\left(X_\mathrm{u} e^{-\beta \nu}+\mathcal{Z}_+e^{\beta c}+\mathcal{Z}_- e^{-\beta c} \right)\label{cwork 2}
\end{equation}
where 
\begin{eqnarray}
  \mathcal{Z}_\pm &=& \sum_{s\in \mathcal X_\mathrm{\pm}}e^{-\beta \mathcal{E}_s}
\end{eqnarray}
are the partition functions calculated over the positive and negative bounded partitions respectively. In Supplementary Notes 3 and 4 we find that in the optimal work extraction and state formation protocols the final / initial state of the system is the Gibbs state, which is diagonal in the energy eigenbasis. Hence equations \eqref{cwork} and \eqref{cform} give the optimal work for arbitrary quantum states (See Supplementary Note 1). For a two-level quantum system $d=2$, the work content can be expressed succinctly as 
\begin{eqnarray}
  \label{simple work 2}
  && W^{(c)}(\rho) = 
  \\ \nonumber
  &&\begin{cases}
   \frac 1 \beta \log\mathcal{Z}'
   -\frac 1 \beta \log\! \left[ 
    e^{+\beta c}\left(1+e^{-\beta 
    (\mathcal{E} + c /x_1)}\right) 
    \right] 
    &  \mbox{if } c<\xi
    \\
    \frac 1 \beta
    \log\mathcal{Z}'-\frac 1 \beta
    \log\! \left[ 
    e^{-\beta c}\left(1+e^{-\beta 
    (\mathcal{E}- c/x_1)}\right) 
    \right] 
    & \mbox{if } c>-\xi
    \\
    \frac 1 \beta
    \log\mathcal{Z'}+ F(\rho)
    & \mbox{otherwise }
  \end{cases}
\end{eqnarray}
where
\begin{equation}
  \xi = \frac 1 \beta \ln(1\!-\!x_1) -F(\rho)\ .
\end{equation}
Without loss of generality we have assumed above $x_1 \geq x_2$ and we define $\mathcal E = \mathcal E_1 - \mathcal E_2$.
%Where $\rho= x\ketbra{0}{0}+(1\!-\!x)\ketbra{1}{1}$, $x\geq 1/2$ and $\mathcal{H}_\text{S}=\mathcal{E} \ketbra{0}{0}$ (note we can write any qubit state in this way, taking $\mathcal{E}<0$ if necessary). 
The above expression derives from \eqref{cwork 2} and the partitioning algorithm given in Supplementary Note 3.\\

\indent Returning to the $c$-bounded work of formation, \eqref{cform} and the algorithm for finding the state space partition giving the $c$-bounded work of formation \eqref{cform} is detailed in Supplementary Note 4. Note that the hamiltonian of $\rho$, the final state of the system, is given by $H'_\text{S}=\sum_j \mathcal{E}'_j \ketbra{j}{j}$. The work of formation for a two level system can be succinctly stated as 

\begin{equation}
W^{(c)}_\text{F}(\rho)=\begin{cases}
\frac{1}{\beta}\log\mathcal{Z}-F(\rho) \, , \,\quad \quad \quad\, \,\, \,c\geq \xi\\
\frac{1}{\beta}\log \left(x\,e^{\beta\mathcal{E}'}\mathcal{Z} \right)+c\frac{1-x}{x}\, , \, c<\xi
\end{cases} \label{formation simple}
\end{equation}
where $\rho=\! x \ketbra{0}{0}\!+(1\!-\!x)\ketbra{1}{1}$, $x\geq 1/2$ and $H_\text{S}=\mathcal{E}\ketbra{0}{0}$ is the Hamiltonian of the initial Gibbs state. \\

\begin{figure}[h!]
\includegraphics[scale=0.6]{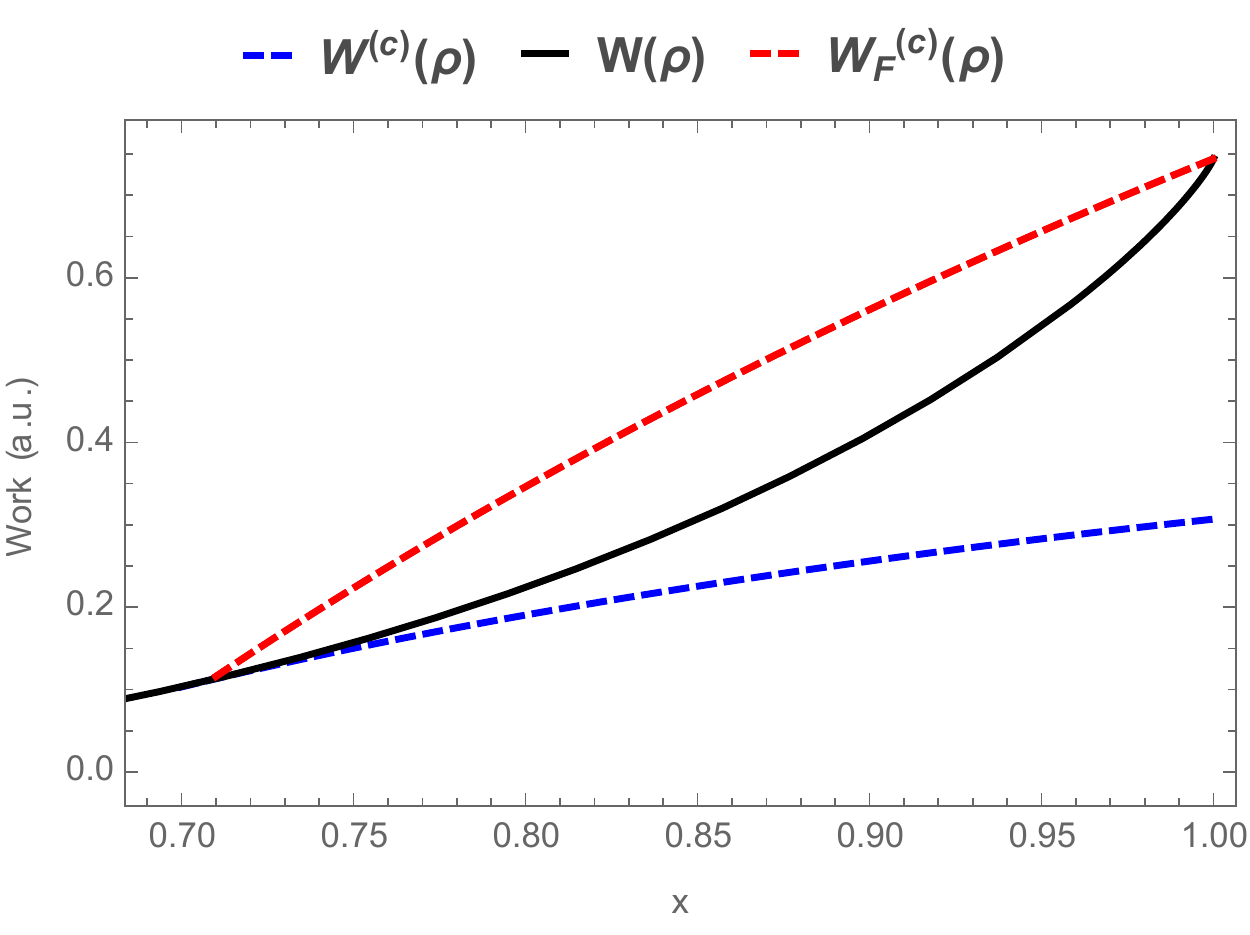}
\caption{How $c$-bounded work extraction and state formation varies further from equilibrium. Figure shows the unbounded work $W(\rho )$, the $c$-bounded work content $W^{(c)}(\rho )$ and the $c$-bounded work of formation $W^{(c)}_F(\rho )$ for the state $\rho= x \ketbra{0}{0} + (1-x)\ketbra{1}{1}$ with Hamiltonian $H_s = \mathcal{E}\ketbra{0}{0}$ with $\beta=1$, $\mathcal{E}=0.1$ and $c=0.7$. Note that despite choosing a $c$-bound that allows for fluctuations of the order of the maximal work that can be extracted from the pure state $x=1$, in general we can extract much less that this amount. There is a discontinuity in $W^{(c)}(\rho )$ at $x=0$ where we recover $W^{(c)}(\rho ) = W^{(\infty)}(\rho)$. Notice also that closer to the thermal state the dissipation (difference between the $W^{(c)}(\rho )$ or $W^{(c)}_F(\rho )$ and $W^{(\infty)}(\rho )$) is greater for state formations than work extraction, and this reverses as the state moves further from the Gibbs state.}\label{fig1}
\end{figure}

\begin{figure}[h!]
\includegraphics[scale=0.6]{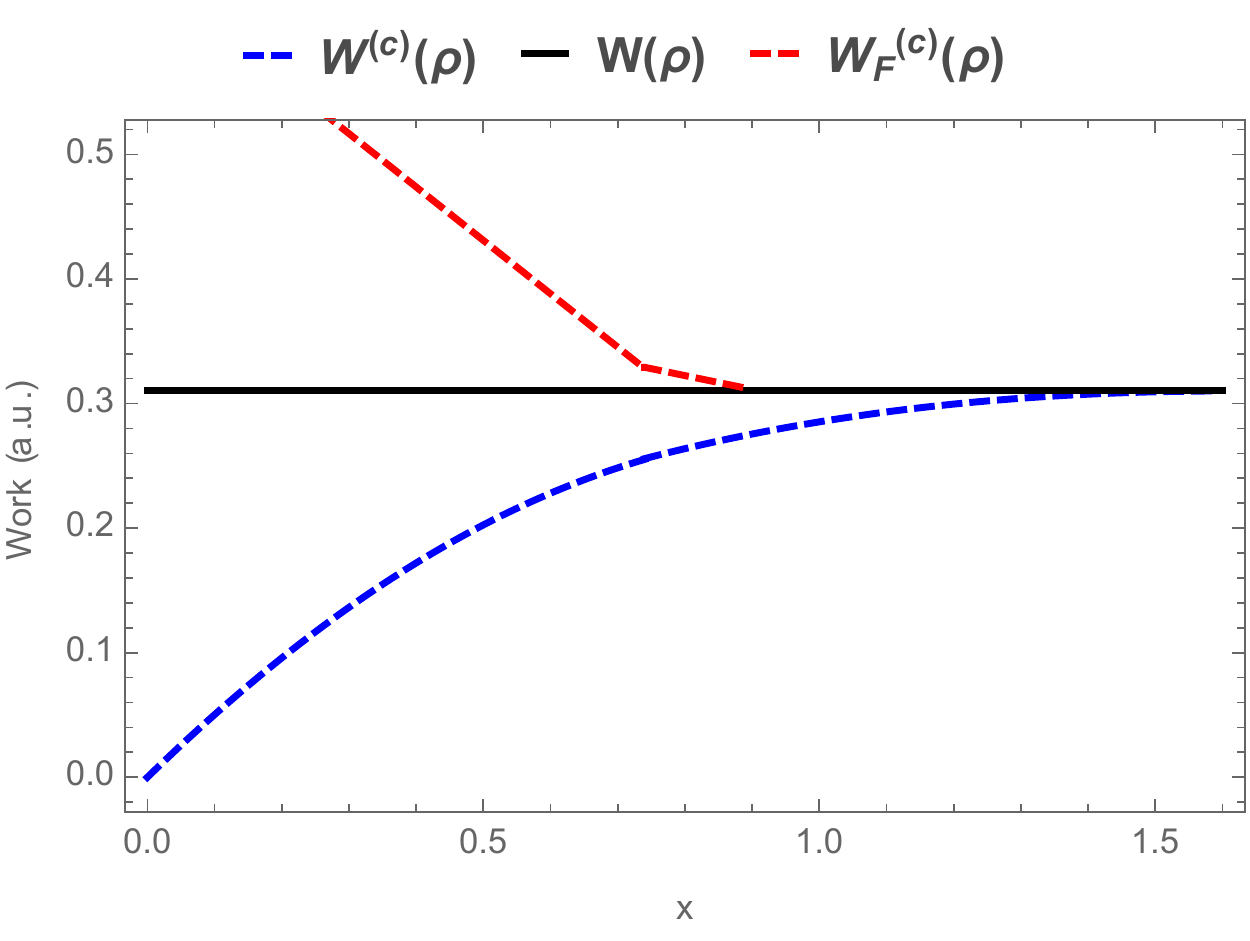}
\caption{How $c$-bounded work extraction and state formation varies with $c$ for a fixed initial / final state. figure shows $W^{(c)}_F (\rho )$ and $W^{(c)} (\rho )$ v.s. $c$ for the trit state $\rho = 0.7 \ketbra{0}{0}+0.2\ketbra{1}{1} +0.1 \ketbra{2}{2}$ with Hamiltonian $H_s=\mathcal{E}_1\ketbra{0}{0}+\mathcal{E}_2\ketbra{1}{1}$ with $\mathcal{E}_1 = 0.1$, $\mathcal{E}_2=0.2$. $\beta$ is set to 1. For small $c$ the dissipation $|W^{(c)}(\rho )-W^{(\infty)}(\rho )|$ for the formation protocol is greater than for the extraction protocol, and for large $c$ the relationship is inverted. Note that for $c>0.9$ it possible to thermodynamically reversibly prepare state $\rho$ but not to thermodynamically reversibly extract work from it. } \label{fig2}
\end{figure}

We now summarise some of the properties of the $c$-bounded work.
\begin{theorem}
The $c$-bounded work is related to the non-fluctuating work by inequalities
\begin{eqnarray}
W^{(c)}(\rho) &\leq& W^{(0)} (\rho) + c \label{bounds inequality}\\
W_\text{F}^{(c)}(\rho) &\geq& W_\text{F}^{(0)} (\rho) - c
\end{eqnarray}
becoming strict inequalities for $c>0$
\end{theorem}

\noindent \emph{Proof:} See Supplementary Note 5.\\

These inequalities imply a fundamental trade-off between work and fluctuations. In order to do better than single shot work extraction / state formation, our work must have fluctuations that are greater than the increase in work / decrease in work cost with respect to the deterministic work. In Supplementary Note 5 we show that the $c$-bounded work distributions that give $W^{(c)}(\rho)$ and $W^{(c)}_F(\rho)$ obey the Jarzinski equality \cite{jarzynski1997nonequilibrium}. In Supplementary Note 5 we prove that $\lim\limits_{c\rightarrow 0}W^{(c)}(\rho) = W^{(0)}(\rho)$ and similarly for $W_\text{F}^{(c)}(\rho)$.\\ 
\indent For the interested reader, it is simple to see that for any finite $c$ a partial ordering of the states w.r.t work emerges. A simple way to observe this is to choose a qubit state $\rho$ with Hamiltonian $H_\text{S}$ and a thermal qubit state $\gamma$ with Hamiltonian $H'_\text{S}\neq H_\text{S}$ such that neither state thermo-majorizes the other (see \citep{horodecki2013fundamental} for examples). For any two such states there is a value of $c$ below which $W^{(c)}(\rho)$ (for $\rho \rightarrow \gamma$) is negative and $W^{(c)}_\text{F}(\rho)$ (for $\gamma \rightarrow \rho$) is positive, i.e. it costs work to perform both the forward and backwards transitions. Note that there are states with the same standard free energy that exhibit a partial order for finite $c$. Allowing the weight to fluctuate allows us to transition between these states freely. This is an example of how weight is not just a resource for extracting additional fluctuating work but in accommodating dynamics, even when on average its displacement remains zero. It is an interesting open question to determine how much we must allow work to fluctuate to allow a transition $\rho\rightarrow \rho'$ to be achieved without costing work.\\
%For example, consider the qubit states\\
%\begin{align*}
%\left(\rho, \, \mathcal{H}_\text{S} \right)&=\left(x\ketbra{0}{0}+(1-x)\ketbra{1}{1},\, \mathcal{E}\ketbra{0}{0} \right) \\
%\left(\sigma, \, \mathcal{H}'_\text{S} \right)&=(\frac{e^{-\beta \mathcal{E}'}}{\mathcal{Z}'}\ketbra{0}{0}+\frac{1}{\mathcal{Z}'}\ketbra{1}{1},\, \mathcal{E}'\ketbra{0}{0} ) 
%\end{align*}
%The transformation $\left(\rho, \, \mathcal{H}_\text{S} \right)\rightarrow\left(\sigma, \, \mathcal{H}'_\text{S} \right)$, i.e. work extraction with a Hamiltonian quench, extracts $c$-bounded work $W^{(c)}(\rho)=\beta^{-1}\log\mathcal{Z}'-F^{(c)}_\text{d=2}(\rho)$, where $F^{(c)}_\text{d=2}(\rho)$ is given in \eqref{simple work 2}. The reverse process is forming state $\left(\rho, \, \mathcal{H}_\text{S} \right)$ from the thermal state $\left(\sigma, \, \mathcal{H}'_\text{S} \right)$, including a Hamiltonian quench. Now $\mathcal{H}_\text{S}'$ is the initial Hamiltonian, and the work cost is given by \eqref{formation simple}. Let $\mathcal{E}=0$. \\

\noindent \textbf{Qubit Carnot engine.} In this section we find the $c$-bounded Carnot efficiency for a qubit carnot engine model, i.e. the maximal efficiency the qubit engine can reach given that fluctuations in the work it produces are bounded by $c$. We use the same single qubit engine model as described in \cite{skrzypczyk2014work}. The engine operates by moving a qubit $\rho$ with Hamiltonian $H_s=\mathcal{E}\ketbra{0}{0}$ between two baths of inverse temperature $\beta_H$ and $\beta_C$, with $\beta_H<\beta_C$. The qubit has state $\rho_{H,C} = \mathcal{Z}_{H,C}^{-1}e^{-\beta_{H,C} \mathcal{E}}\ketbra{0}{0}+\mathcal{Z}_{H,C}^{-1}\ketbra{1}{1}$ when in equilibrium with the hot / cold bath, where $\mathcal{Z}_{H,C}= 1+e^{-\beta_{H,C}\mathcal{E}}$. The engine cycle begins with the qubit in thermal equilibrium with the cold bath. In the first half of the cycle it is then placed in contact with the hot bath and work is extracted. In the second step of the cycle the qubit is returned to the cold bath and work is extracted a second time. In Supplementary Note 6 we show that, in the case that fluctuations are not bounded, it is possible to reach Carnot efficiency with this engine, as shown in \cite{skrzypczyk2014work}.

\begin{equation}
\eta_\text{Carnot}=1-\frac{\beta_H}{\beta_C}
\end{equation}
In the case that $c$ is finite, the work extracted in the first half of the cycle is given by 
\begin{equation}
W_1^{(c)}(\rho ) = \begin{cases}
\frac{1}{\beta_H}\log \left(\frac{\mathcal{Z}_H}{\mathcal{Z}_C}\right)+\text{tr}[H_s \rho_C ]\left(\frac{\beta_H-\beta_C}{\beta_H}\right),\, \text{if} \, A>c \\
\frac{1}{\beta_H}\log\left(\frac{\mathcal{Z}_He^{\beta_H c}}{e^{c\beta_H\mathcal{Z}_C}+e^{-\mathcal{E}\beta_H}}\right)\, , \, \text{if}\, A\leq c
\end{cases}
\end{equation}
where 
\begin{equation}
A = 
\frac{\mathcal{E}}{\mathcal{Z}_C}\left(\frac{\beta_C-\beta_H}{\beta_H} \right) \label{cef1}
\end{equation}
if $A\leq c$ we simply extract the difference in free energy between the two thermal states, otherwise we extract the $c$-bounded work of $\rho_C$ in contact with the bath $\beta_H$. Similarly, on the second part of the cycle we extract 
\begin{equation}
W_2^{(c)}(\rho ) = \begin{cases}
\frac{1}{\beta_C}\log \left(\frac{\mathcal{Z}_C}{\mathcal{Z}_H}\right)+\text{tr}[H_s \rho_H ]\left(\frac{\beta_C-\beta_H}{\beta_H}\right),\, \text{if} \, B>c \\
\frac{1}{\beta_C}\log\left(\frac{\mathcal{Z}_Ce^{-\beta_C c}}{e^{-c\beta_C\mathcal{Z}_H}+e^{-\mathcal{E}\beta_C}}\right)\, , \, \text{if}\, B\leq c
\end{cases}
\end{equation}
where
\begin{equation}
B = 
\frac{\mathcal{E}}{\mathcal{Z}_H}\left(\frac{\beta_C-\beta_H}{\beta_C} \right) \label{cef2}
\end{equation}
Note that satisfying \eqref{cef2} implies that \eqref{cef1} is also satisfied, therefore breaking inequality \eqref{cef1} is the condition for achieving Carnot efficiency in this model. Also note that $B$ gives the minimum worst case fluctuation of the work extracted by this engine when operating thermodynamically reversibly. The efficiency is given by the ratio of the heat flow from the hot bath to the total work extracted in the cycle. The heat flow from the hot bath is found by applying the $1^\text{st}$ law of thermodynamics, $Q_H = \Delta \langle \mathcal{E}\rangle  (\rho_H \rightarrow \rho_C) + W_1^{(c)}$ where $ \Delta \langle \mathcal{E}\rangle  (\rho_H \rightarrow \rho_C)$ is the change in the systems internal energy in the first part of the cycle. Therefore the $c$-bounded efficiency of the engine is given by 
\begin{equation}
\eta^{(c)} =\frac{W^{(c)}_1+W^{(c)}_2}{\Delta U + W^{(c)}_1}
\end{equation} 
\begin{figure}[h!]
\includegraphics[scale=0.6]{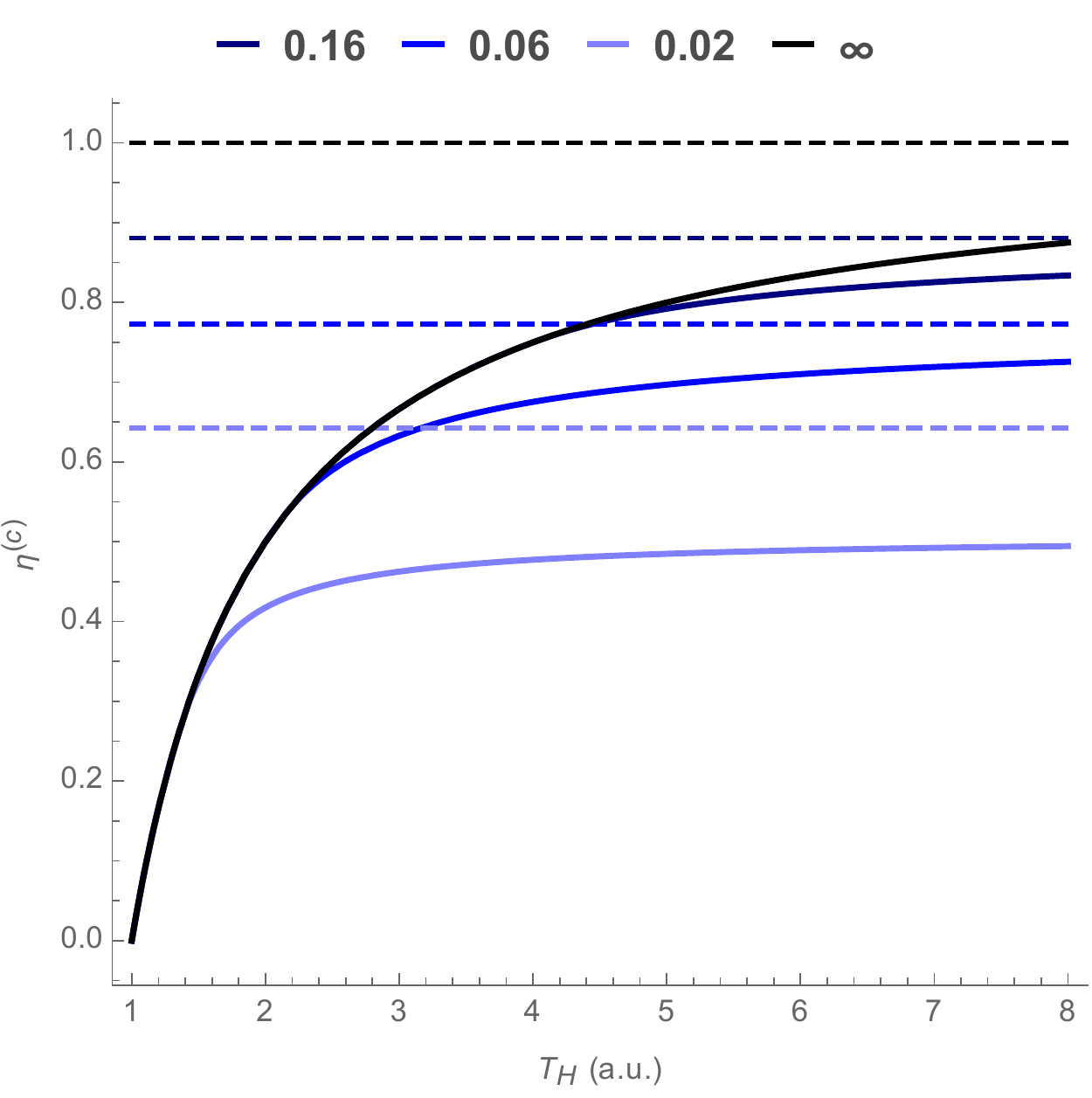}
\caption{The corrected carnot efficiency for various values of $c$ for a single qubit heat engine. Figure shows the $c$-bounded Carnot efficiencies vs. $T_H$ for single qubit engine with gap $\mathcal{E}=0.1$ and $T_C=1$. The black line gives the unbounded Carnot efficiency. The Blue lines give the $c$-bounded efficiencies with there corresponding $c$'s marked on the figure. The dashed lines give the maximum attainable efficiency in the limit of asymptotic temperature difference between $\beta_C/\beta_H\rightarrow \infty$}
\end{figure}
In the case that $c\rightarrow\infty$, we recover the Carnot efficiency, which is bounded from above by 1, i.e. we recover unit efficiency in the limit that $\beta_C/\beta_H\rightarrow\infty$. For any finite $c$ this is no longer the case, with the maximal efficiency give by \\
\begin{equation}
\eta^{(c)}_\text{max}=\lim\limits_{\beta_C/\beta_H \rightarrow\infty} \eta^{(c)}= 1-\frac{\mathcal{E}}{2(2c+\mathcal{E})}
\end{equation}
This gives an upper limit on the efficiency of the single qubit engine protocol described above, which is dependent only on the Hamiltonian of the qubit and the parameter $c$.  \\
\indent As $\mathcal{E}\rightarrow 0$, $\eta^{(c)}_\text{max}\rightarrow 1$, but the work extracted tends to zero as the Gibbs states associated with the two bath temperatures become indistinguishable. For $c\rightarrow 0$ we get that $\eta^{(c)}_\text{max}$ is bounded from below by $1/2$, but at $c=0$ no work can be extracted as the thermal states are of full rank, giving $\eta^{(0)}_\text{max}=0$. Therefore we find that, although this engine cannot run at non-zero efficiency in the single-shot regime, if we allow for arbitrarily small fluctuations it is possible in principle to reach a maximum efficiency greater that $1/2$ (although the fluctuations in work will still be of the order of the work extracted, given inequality \eqref{bounds inequality}). Similar results relating to the single-shot regime are discussed in \cite{woods2015maximum}.   \\

\section*{Discussion}

\noindent In this article we have derived tight bounds on the minimal fluctuations in work associated with thermodynamically reversible protocols, for which the average work is given by difference in free energy between initial and final states. We have found that thermodynamically reversible protocols have fluctuations that diverge in size as the relative athermality of initial or final states increases.\\
\indent Motivated by this we have presented a framework for computing the work associated with a thermodynamic process under arbitrary convex bounds. We have derived the $c$-bounded work content and work of formation of arbitrary quantum states, which can be understood as modified free energies that interpolate between the standard and single shot free energies. By exploring this new territory, we have found that the phenomenology of single-shot thermodynamics, namely thermodynamic irreversibility and a partial order of states with respect to work, are to some extent present for any finite $c$. Furthermore we have found that it is impossible to extract more that the deterministic work content of a system without necessitating fluctuations that are greater than the gain in work (and similarly for the work cost of state formation). One potential avenue for extending these results would be to consider a more general definition of $c$-bounded work that includes a small probability of failure in the work extraction process (as proposed in \cite{aaberg2013truly})\\
\indent An interesting open question is to what extent we must allow work to fluctuate in order to allow for a given state transformation. Answering this question would require the extension of the results presented in this article to processes with arbitrary initial and final states, including the case where both initial and final states contain coherences between energy levels. In Supplementary Note 1C we show that, under the assumption that the protocol is independent on the position state of the weight, the ``coherence modes'' \cite{lostaglio2015description,lostaglio2015quantum,cwiklinski2015limitations} evolve independently under the action of the thermal map. This lays the ground for future investigations into how the presence of coherences affects the allowed thermal operations in the case that work is allowed to fluctuate. \\
\indent Finally, we have used the $c$-bounded work to study how bounding work fluctuations affects the efficiency of a single qubit nano-engine, and have derived an upper bound on efficiency of this engine that depends only on $c$ and the engine's Hamiltonian, establishing a fundamental trade-off between a the engines efficiency and the fluctuations in the work it produces. This opens the door to correcting the efficiency for general thermodynamic protocols, taking into account the fragility of realistic machines that cannot tolerate large fluctuations in work.\\ 
\indent Given that there are many thermal engine models that can reach Carnot efficiency in the case the fluctuations in work are unbounded, it would be of interest to determine the optimal engine with respect to minimizing fluctuations in work whilst maximising efficiency or the power produced. Furthermore, it is well known that in the thermodynamic limit the relative size of fluctuations in work to the average tends to zero. It would be of interest to determine if it is possible to design engines operating far from the thermodynamic limit that achieve a similar quasi-deterministic work output with non-zero power. For example it could be possible, through clever choice of the working system Hamiltonian, or by controlling interactions between a small number of systems that constitute the working system, to find engine models that achieve quasi-deterministic work output without needing to take the thermodynamic limit. Further work in this direction would provide invaluable insights for designing realistic nano-engines that are robust to fluctuations in work. \\
%\indent Our framework can in principle be employed to calculate the work cost of a protocol under any reasonable bounds, for example for bounding higher moments of the work distribution or for fixed bounds in the case of charging a battery with capacity $c_+$ and charge $c_-$ or work extraction from finite baths. The utility of our framework is twofold, to find ultimate bounds on energy transfer in thermodynamic protocols and to derive bespoke free energies for protocols operating far from the thermodynamic limit with realistic resources such as fragile engines, bounded work storage systems, and finite baths. \\ 

\section*{Methods}

\noindent \textbf{Proofs and derivations.} The proofs are contained in he Supplementary information. In Supplementary note 1 we address the preliminaries and framework within which we derive our proofs, including the framework of thermal operations with fluctuating work, changing Hamiltonians, coherences and reducing quantum the protocols to classical protocols, and proofs that our framework is both general and optimal. In Supplementary note 2 we derive the form of the work optimal protocols with unbounded fluctuations, and derive Result 1. In Supplementary note 3 we derive the c-bounded work content, and the corresponding state partitioning algorithm. In Supplementary note 4 we derive the c-bounded work of formation, and the corresponding state partitioning algorithm. In Supplementary note 5 we show that the standard and min free energy can be recovered in the limits $c\rightarrow \infty$ and $c\rightarrow 0$ respectively, and derive the trade-off bounds relating the optimal work to the size of the worst-case fluctuations. In Supplementary note 6 we derive the c-bounded Carnot for the single qubit heat engine model. \\

\noindent \textbf{Data availability}\\ 
\noindent Data sharing is not applicable to this article, as no data sets were generated or analysed during this study. \\

\begin{comment}
\bibliographystyle{unsrt}

\bibliography{preprintbibliography}

\begin{thebibliography}{0}
\expandafter\ifx\csname natexlab\endcsname\relax\def\natexlab#1{#1}\fi
\expandafter\ifx\csname bibnamefont\endcsname\relax
  \def\bibnamefont#1{#1}\fi
\expandafter\ifx\csname bibfnamefont\endcsname\relax
  \def\bibfnamefont#1{#1}\fi
\expandafter\ifx\csname citenamefont\endcsname\relax
  \def\citenamefont#1{#1}\fi
\expandafter\ifx\csname url\endcsname\relax
  \def\url#1{\texttt{#1}}\fi
\expandafter\ifx\csname urlprefix\endcsname\relax\def\urlprefix{URL }\fi
\providecommand{\bibinfo}[2]{#2}
\providecommand{\eprint}[2][]{\url{#2}}

\end{thebibliography}


\begin{thebibliography}{10}

\bibitem{horodecki2013fundamental}
Micha{\l} Horodecki and Jonathan Oppenheim.
\newblock Fundamental limitations for quantum and nanoscale thermodynamics.
\newblock {\em Nature communications}, 4, 2059, (2013).

\bibitem{skrzypczyk2014work}
Paul Skrzypczyk, Anthony~J Short, and Sandu Popescu.
\newblock Work extraction and thermodynamics for individual quantum systems.
\newblock {\em Nature communications}, 5, 4185 (2014).

\bibitem{gemmer2015single}
J~Gemmer and J~Anders.
\newblock From single-shot towards general work extraction in a quantum
  thermodynamic framework.
\newblock {New Journal of Physics},  17.8, 085006 (2015)

\bibitem{aaberg2013truly}
Johan {\AA}berg.
\newblock Truly work-like work extraction via a single-shot analysis.
\newblock {\em Nature communications}, 4, 1925 (2013).

\bibitem{dahlsten2011inadequacy}
Oscar~CO Dahlsten, Renato Renner, Elisabeth Rieper, and Vlatko Vedral.
\newblock Inadequacy of von neumann entropy for characterizing extractable
  work.
\newblock {\em New Journal of Physics}, 13(5):053015, (2011).

\bibitem{linden2010small}
Noah Linden, Sandu Popescu, and Paul Skrzypczyk.
\newblock How small can thermal machines be? the smallest possible
  refrigerator.
\newblock {\em Physical review letters}, 105(13):130401, (2010).

\bibitem{ryan2008spin}
CA~Ryan, O~Moussa, J~Baugh, and R~Laflamme.
\newblock Spin based heat engine: demonstration of multiple rounds of
  algorithmic cooling.
\newblock {\em Physical review letters}, 100(14):140501, (2008).

\bibitem{quan2006maxwell}
HT~Quan, YD~Wang, Yu-xi Liu, CP~Sun, and Franco Nori.
\newblock Maxwell’s demon assisted thermodynamic cycle in superconducting
  quantum circuits.
\newblock {\em Physical review letters}, 97(18):180402, (2006).

\bibitem{brandao2015second}
Fernando Brand{\~a}o, Micha{\l} Horodecki, Nelly Ng, Jonathan Oppenheim, and
  Stephanie Wehner.
\newblock The second laws of quantum thermodynamics.
\newblock {\em Proceedings of the National Academy of Sciences},
  112(11):3275--3279, (2015).

\bibitem{masanes2014derivation}
Lluis Masanes and Jonathan Oppenheim.
\newblock A derivation (and quantification) of the third law of thermodynamics.
\newblock {\em Preprint at https://arxiv.org/abs/1412.3828}, (2014).

\bibitem{gallego2015defining}
R~Gallego, J~Eisert, and H~Wilming.
\newblock Defining work from operational principles.
\newblock {\em Preprint at https://arxiv.org/abs/1504.05056}, (2015).

\bibitem{lostaglio2015description}
Matteo Lostaglio, David Jennings, and Terry Rudolph.
\newblock Description of quantum coherence in thermodynamic processes requires
  constraints beyond free energy.
\newblock {\em Nature communications}, 6, 6383, (2015).

\bibitem{lostaglio2015quantum}
Matteo Lostaglio, Kamil Korzekwa, David Jennings, and Terry Rudolph.
\newblock Quantum coherence, time-translation symmetry, and thermodynamics.
\newblock {\em Physical Review X}, 5(2):021001, (2015).

\bibitem{alhambra2016second}
{\'A}lvaro~M Alhambra, Lluis Masanes, Jonathan Oppenheim, and Christopher
  Perry.
\newblock The second law of quantum thermodynamics as an equality.
\newblock {\em Preprint at https://arxiv.org/abs/1601.05799}, (2016).

\bibitem{aaberg2014catalytic}
Johan Aberg.
\newblock Catalytic coherence.
\newblock {\em Physical review letters}, 113(15):150402, (2014).

\bibitem{perarnau2016quantum}
Marti Perarnau-Llobet, Elisa Baumer, Karen~V Hovhannisyan, Marcus Huber, and
  Antonio Acin.
\newblock Quantum fluctuations of work and generalised quantum measurements.
\newblock {\em Preprint at https://arxiv.org/abs/1606.08368}, (2016).

\bibitem{talkner2016aspects}
Peter Talkner and Peter Hanggi.
\newblock Aspects of quantum work.
\newblock {\em Physical Review E}, 93(2):022131, (2016).

\bibitem{gallego2015thermodynamic}
R~Gallego, J~Eisert, and H~Wilming.
\newblock Thermodynamic work from operational principles.
\newblock {\em Preprint at https://arxiv.org/abs/1504.05056}, (2015).

\bibitem{murashita2015absolute}
Y{\^u}to Murashita and Masahito Ueda.
\newblock Absolute irreversibility resolves the gibbs paradox.
\newblock {\em Preprint at https://arxiv.org/abs/1506.04468}, (2015).

\bibitem{murashita2014nonequilibrium}
Y{\^u}to Murashita, Ken Funo, and Masahito Ueda.
\newblock Nonequilibrium equalities in absolutely irreversible processes.
\newblock {\em Physical Review E}, 90(4):042110, (2014).

\bibitem{jarzynski1997nonequilibrium}
Christopher Jarzynski.
\newblock Nonequilibrium equality for free energy differences.
\newblock {\em Physical Review Letters}, 78(14):2690, (1997).

\bibitem{woods2015maximum}
Mischa~P Woods, Nelly Ng, and Stephanie Wehner.
\newblock The maximum efficiency of nano heat engines depends on more than
  temperature.
\newblock {\em Preprint at https://arxiv.org/abs/1506.02322}, (2015).

\bibitem{cwiklinski2015limitations}
Piotr {\'C}wikli{\'n}ski, Micha{\l} Studzi{\'n}ski, Micha{\l} Horodecki, and
  Jonathan Oppenheim.
\newblock Limitations on the evolution of quantum coherences: Towards fully
  quantum second laws of thermodynamics.
\newblock {\em Physical review letters}, 115(21):210403, (2015).

\end{thebibliography}
\end{comment}

\bigskip

\noindent \textbf{Acknowledgements} \\
\noindent JR would like to thank Jochen Gemmer, Martí Perarnau-Llobet, Alvaro Alhabra and Chris Perry for discussions. JR is supported by ESPRC.\\

\noindent \textbf{Author contributions}\\
\noindent Both authors contributed equally to this work

\noindent \textbf{Author competing financial interests}\\
\noindent The authors declare no competing financial interests

\newpage

\section{Supplementary note 1}  

\subsection{Thermal operations with fluctuating work}

\noindent In this supplementary note we introduce the framework that we will use to derive $W^{(c)} (\rho )$. The work extraction protocol is performed using a system $\rho$, a infinite thermal bath and a work system or weight. \\

First, let us characterize the type of process/operation that we consider, which we refer to as thermal operations with fluctuating work. Our setting consists of a system with Hamiltonian $H_{\rm S}$, a bath with Hamiltonian $H_{\rm B}$ initially in the thermal state $\rho_\mathrm{B} = \frac 1 {Z_{\rm B}} \e^{-\beta H_\mathrm{B}}$, and an ideal weight with Hamiltonian $H_\mathrm{W}=  \int_\mathbb{R} dx\, x |x\rangle \! \langle x|$, where the orthonormal basis $\{|x\rangle, \forall\, x\in \mathbb R\}$ represents the position of the weight. 
% The operations we consider will allow for the Hamiltonian to change as we shall see in Subsection \ref{ss:toH}.
Any joint transformation of system, bath and weight is represented by a Completely Positive Trace Preserving (CPTP) map $\Gamma_\mathrm{SBW}$ satisfying the following conditions:
\begin{description}
  \item[Microscopic reversibility (Second Law):] It has an (CPTP) inverse $\Gamma_\mathrm{SBW}^{-1}$, which implies unitarity $\Gamma_\mathrm{SBW} (\rho_\mathrm{SBW}) = U\rho_\mathrm{SBW} U^\dagger$.
    
  \item[Energy conservation (First Law):] 
  $[U,H_\mathrm{S} +H_\mathrm{B} +H_\mathrm{W}] =0$.

  \item[Independence from the ``position" of the weight:]
  The unitary commutes with the translations on the weight $[U,\Delta_\mathrm{W}] = 0$. The generator of the translations $\Delta_\mathrm{W}$ is canonically conjugated to the position of the weight $[H_\mathrm{W}, \Delta_\mathrm{W}] = \i$.

  \item[Classicality of work:] Before and after applying the global map $\Gamma_\mathrm{SBW}$ the position of the weight is measured, obtaining outcomes $|x\rangle$ and $|x+w\rangle$ respectively. The joint transformation of the system and work random variable $w$ is given by the map
\begin{equation}
  \label{TO}
  \Lambda (\rho_\mathrm{S},w)
  =
  \int_\mathbb{R}\! dx\, {\rm tr}_\mathrm{BW}\! 
  \left[ Q_{x+w}\,
  U \left( \rho_\mathrm{S}\, \rho_\mathrm{B}\,
  Q_x \rho_\mathrm{W} Q_x
  \right) U^\dagger \right]\ ,
\end{equation}
where $Q_x = |x\rangle\! \langle x|$ is a weight eigen-projector.

%  \item[Independence from the state of the weight:] The map $\Gamma_\mathrm{S} (\rho_\mathrm{S},w)$ is independent of the initial state of the weight $\rho_\mathrm{W}$.
  
\end{description}
Condition $[U,\Delta_\mathrm{W}] = 0$  implies that the reduced map on system and bath is a mixture of unitaries (Theorem 1 in \cite{masanes2014derivation}). Hence, this transformation can never decrease the entropy of system and bath, which guarantees that the weight is not used as a source of non-equilibrium. 
%It is important to mention that these constraints allow for processes that exploit the coherence of the weight, as in~\cite{Aberg coherence catalysis}.
%
Note that after integrating over the work varable
\begin{equation}
  \int_\mathbb{R}\! dw\, 
  \Lambda (\rho_\mathrm{S},w)
  = {\rm tr}_\mathrm{BW}\! \left[U 
  \left( \rho_\mathrm{S}\, \rho_\mathrm{B}\, \Theta [\rho_\mathrm{W}] \right) 
  U^\dagger \right]
  = \Gamma_\mathrm{S} (\rho_\mathrm{S}) \ ,
\end{equation}
we obtain the reduced CPTP map for the system. The state $\Theta[\rho_\mathrm{W}]  = \int_\mathbb{R}\! dx\, Q_x \rho_\mathrm{W} Q_x$ is the energy-diagonal version of $\rho_\mathrm{W}$. But $\Lambda (\rho_\mathrm{S},w)$ being independent of $\rho_\mathrm{W}$, we can choose the initial state to be diagonal $\Theta [\rho_\mathrm{W}]$.
Also, note that when tracing the system
\begin{equation}
  {\rm tr}_\mathrm{S} \Lambda (\rho_\mathrm{S},w)
  = P(w)\ ,
\end{equation}
we obtain the probability distribution of the work generated in the transformation.\\

\subsection{Thermal operations with non-constant Hamiltonian}
\label{ss:toH}

Thermal operations are general enough to include the case where the initial Hamiltonian of the system $H_\mathrm{S}$ is different than the final one $H'_\mathrm{S}$. This is done by including an additional qubit $X$ which plays the role of a switch (as in \cite{horodecki2013fundamental,aberg2016fully}). Now the total Hamiltonian is
\begin{equation}
  H = H_\mathrm{S}\otimes |0\rangle _\mathrm{X}%{\rm switch} 
  \langle 0| + H'_\mathrm{S} \otimes |1\rangle_\mathrm{X}%{\rm switch} 
  \langle 1| +H_\mathrm{B} +H_\mathrm{W}
\ ,
\end{equation}
and energy conservation reads 
$[V,H] =0$, where $V$ is the global unitary when we include the switch.
We impose that the initial state of switch is $|0\rangle_\mathrm{X}$ and the global unitary $V$ performs the switching
\begin{equation}
  \label{SS}
V \left( 
\rho_\mathrm{SBW} \otimes |0\rangle_\mathrm{X} \langle 0|\right)  V^\dagger = \rho'_\mathrm{SBW} \otimes |1\rangle_\mathrm{X} \langle 1| 
\ ,
\end{equation}
for any $\rho_\mathrm{SBW}$.
This implies
\begin{equation}
  V=  U
  \otimes |1\rangle_\mathrm{X} \langle 0| +
  \tilde U
  \otimes |0\rangle_\mathrm{X} \langle 1|
  \ ,
\end{equation}
where $U$ and $\tilde U$ are unitaries on system, bath and weight. Condition $[V,H] =0$ implies
\begin{equation}
  \label{EC}
  U (H_\mathrm{S}+H_\mathrm{B}+H_\mathrm{W}) =
  (H'_\mathrm{S}+H_\mathrm{B}+H_\mathrm{W}) U
  \ .
\end{equation}
Therefore, the reduced map on system, bath and weight can be written as
\begin{equation}
  \label{rmm}
  \Gamma_\mathrm{SBW} (\rho_\mathrm{SBW}) =
  U \rho_\mathrm{SBW} 
  U^\dagger\ ,
\end{equation}
where the unitary $U$ does not necessarily commute with $H_\mathrm{S}+H_\mathrm{B}+H_\mathrm{W}$ nor $H'_\mathrm{S}+H_\mathrm{B}+H_\mathrm{W}$ but satisfies~\eqref{EC}.

\onecolumngrid

\subsection{Reducing the quantum problem to the classical case}

Let us show that the equation connecting the initial state of the system $\rho_\mathrm{S}$ with the final one conditioned on work $w$,
\begin{equation}
  \rho'_{\mathrm S|w} = \frac 1 {P(w)} 
  \Lambda (\rho_\mathrm{S},w)\ ,
\end{equation}
decouples in the diagonal part and the other energy modes \citep{lostaglio2015description,lostaglio2015quantum,cwiklinski2015limitations}.
The map $\Theta_\alpha$ defined as
\begin{equation}
  \Theta_\alpha [\rho_\mathrm{S}] = 
  \int_\mathbb{R}\! dt\, 
  \e^{i\alpha t}\,
  \e^{iH_\mathrm{S} t} \rho_\mathrm{S} \e^{-iH_\mathrm{S} t}
  \ ,  
\end{equation}
projects the state $\rho_\mathrm{S}$ onto 
the $\alpha$-energy mode of $H_\mathrm{S}$. When $\alpha=0$ it projects the state onto 
its diagonal, when written in the energy eigenbasis. And in general, it projects the state onto all the terms $|s_1\rangle \!\langle s_2|$ such that $\mathcal E_{s_1} - \mathcal E_{s_2} = \alpha$.
Using constraint~\eqref{EC} and identities $\e^{itH_\mathrm{W}} Q_x =\e^{itx}$ and $[H_\mathrm{B}, \rho_\mathrm{B}]=0$  we obtain
\begin{eqnarray}
  \nonumber
  \Theta_\alpha [\Lambda (\rho_\mathrm{S},w)] 
  &=& 
  \int_\mathbb{R}\! dt\, 
  \e^{i\alpha t}\,
  \e^{iH'_\mathrm{S} t}\, \Lambda (\rho_\mathrm{S},w)\, \e^{-iH'_\mathrm{S} t}
  \\ \nonumber &=&  
  \int_\mathbb{R}\! dt\, dx\, 
  \e^{i\alpha t}\,
  {\rm tr}_\mathrm{BW}\! 
  \left[ \e^{i(H'_\mathrm{S}+H_\mathrm{B}+H_\mathrm{W}) t} Q_{x+w}\,
  U \left( \rho_\mathrm{S}\, \rho_\mathrm{B}\,
  Q_x \rho_\mathrm{W} Q_x
  \right) U^\dagger \e^{-i(H'_\mathrm{S}+H_\mathrm{B}+H_\mathrm{W}) t}\right]
  \\ \nonumber &=&  
  \int_\mathbb{R}\! dt\, dx\, 
  \e^{i\alpha t}\,
  {\rm tr}_\mathrm{BW}\! 
  \left[ Q_{x+w}\,
  U \left( \e^{i(H_\mathrm{S}+H_\mathrm{B}+H_\mathrm{W}) t} \rho_\mathrm{S}\, \rho_\mathrm{B}\,
  Q_x \rho_\mathrm{W} Q_x \e^{-i(H_\mathrm{S}+H_\mathrm{B}+H_\mathrm{W}) t} 
  \right) U^\dagger\right]
  \\ \nonumber &=&  
  \int_\mathbb{R}\! dt\, dx\, 
  \e^{i\alpha t}\,
  {\rm tr}_\mathrm{BW}\! 
  \left[ Q_{x+w}\,
  U \left( \e^{iH_\mathrm{S} t} \rho_\mathrm{S} \e^{-iH_\mathrm{S} t}\, \rho_\mathrm{B}\,
  Q_x \rho_\mathrm{W} Q_x  
  \right) U^\dagger\right]
  \\ &=&
  \Gamma_\mathrm{S}(\Theta_\alpha [\rho_\mathrm{S}],w)\ ,
\end{eqnarray}
as claimed above.
This shows that when the initial state is diagonal, $\Theta_\alpha [\rho_\mathrm{S}] = 0$ for all $\alpha \neq 0$, so is the final one; and visa-versa. Thus our results, concerning work extraction and state formation where either the initial or final state diagonal in the energy eigenbasis, are valid in the case that the non-equilibrium state involves coherences between energy eigenstates. \\

The diagonal part of the map is nicely characterized by the conditional probability distribution
\begin{equation}
  t(s',\, w|\, s) = \langle s' |
  \Lambda (|s\rangle\! \langle s|,w)
  |s'\rangle\ ,
\end{equation}
where $|s\rangle$ is the eigenbasis of $\Theta[\rho_\mathrm{S}]$ and $|s'\rangle$ is the eigenbasis of $\Theta[\Lambda (\rho_\mathrm{S},w)]$. 
%Note that when $H'_\mathrm{S}$ is degenerate, the eigenbasis of  $\Theta[\Gamma_\mathrm{S}(\rho_\mathrm{S},w)]$ can depend on $w$, hence we write it explicitly in $|s'_\mathrm{W}\rangle$.
%
Note that the ``dynamics" of the diagonals is a completely classical problem. From now on we use $\rho = \Theta[\rho_\mathrm{S}]$, $\rho' = \Gamma_\mathrm{S}\Theta[\rho_\mathrm{S}]$ for the initial/final states in optimal work extraction / state formation processes.\\
\subsection{Necessary condition for thermal operations}

Using~\eqref{EC} and definitions made above we obtain
\begin{eqnarray}
  \nonumber &&
  \sum_{s} \int_\mathbb{R}\! dw\, t(s',w|s)\, 
  \e^{\beta (\mathcal E_{s'} -\mathcal E_{s} +w)}
  \\ \nonumber &=& 
  \sum_{s}\int_\mathbb{R}\! dw\, dx\, 
  \langle s'|\, \e^{\beta H'_\mathrm{S}}\, 
  {\rm tr}_\mathrm{BW}\! 
  \left[ \e^{\beta H_\mathrm{W}} Q_{x+w}\,
  U \left( \e^{-\beta H_\mathrm{S}}
  |s\rangle\! \langle s|\, \rho_\mathrm{B}\,
  \e^{-\beta H_\mathrm{W}} Q_x \rho_\mathrm{W} Q_x
  \right) U^\dagger \right]
  |s'\rangle
  \\ \nonumber &=& 
  \int_\mathbb{R}\! dx'\, dx\, 
  \langle s'|\, \e^{\beta H'_\mathrm{S}}\, 
  {\rm tr}_\mathrm{BW}\! 
  \left[ \e^{\beta H_\mathrm{W}} Q_{x'}\,
  U \left( \e^{-\beta H_\mathrm{S}}
  \frac {\e^{-\beta H_\mathrm{B}}} {Z_\mathrm{B}}
  \e^{-\beta H_\mathrm{W}} Q_x \rho_\mathrm{W} Q_x
  \right) U^\dagger \right]
  |s'\rangle
  \\ \nonumber &=& 
  \langle s'|\, \e^{\beta H'_\mathrm{S}}\, 
  {\rm tr}_\mathrm{BW}\! 
  \left[ \e^{\beta H_\mathrm{W}}
  \e^{-\beta H'_\mathrm{S}}
  \frac {\e^{-\beta H_\mathrm{B}}} {Z_\mathrm{B}}
  \e^{-\beta H_\mathrm{W}}
  U \Theta[\rho_\mathrm{W}]
  U^\dagger \right]
  |s'\rangle
  \\ \label{last step} &=& 
  {\rm tr}_\mathrm{SBW} \left[ |s'\rangle\! \langle s'|
  \frac {\e^{-\beta H_\mathrm{B}}} {Z_\mathrm{B}}
  U \Theta[\rho_\mathrm{W}]
  U^\dagger \right]
  \ .
\end{eqnarray}
As mentioned above, Theorem 1 in \cite{masanes2014derivation} proves that condition $[U, \Delta_\mathrm{W}] = 0$ implies
\begin{equation}
  \mathrm{tr}\mathrm{W} \left[ U {\mathbb I}_\mathrm{SB} \rho_\mathrm{W} U^\dagger \right] = {\mathbb I}_\mathrm{SB}
\end{equation}
for all $\rho_\mathrm{W}$.
Applying this to~\eqref{last step} we obtain
\begin{equation}
  {\rm tr}_\mathrm{SBW} \left[ 
  |s'\rangle\! \langle s'|
  \frac {\e^{-\beta H_\mathrm{B}}} {Z_\mathrm{B}}
  U \Theta[\rho_\mathrm{W}]
  U^\dagger \right]
  =
  {\rm tr}_\mathrm{SB} \left[ 
  |s'\rangle\! \langle s'|
  \frac {\e^{-\beta H_\mathrm{B}}} {Z_\mathrm{B}}
  \right] = 1
  \ .
\end{equation}
This proves the ``only if" part of the following

\begin{theorem}
  The (classical) map $t(s',w|s)$ comes from a thermal operation 
\begin{equation}
  t(s',w|s) = \langle s' |\Lambda (|s\rangle\! \langle s|,w) |s'\rangle\ ,
\end{equation}
if and only if 
\begin{equation}
  \label{result 1}
  \sum_{s} \int_\mathbb{R}\! dw\, 
  t(s',w|s)\, 
  \e^{\beta (\mathcal E_{s'} -\mathcal E_{s} +w)}
  =1\ ,
\end{equation}
for all $s'$.
\end{theorem}

The ``if" part of the above theorem is proven in Theorem~5 from \cite{alhambra2016second}.

\subsection{Optimal thermal operations have minimal work fluctuations}

Given a process $t(s',w|s)$, satisfying~\eqref{result 1}, the work generated in a particular state transition $s \to s'$ has probability distribution 
\begin{equation}
  \label{Pwss}
  t(w|s,s') = 
  \frac {t(s',w|s)} 
  {t(s'|s)}\ ,
\end{equation}
where $t(s'|s) = \int_\mathbb{R} dw\, t(s',w|s)$ is the value transformation on the system. In general, despite the conditioning on $s,s'$, the distribution~\eqref{Pwss} contains fluctuations on $w$.
In certain setups, a less fluctuating work variable $w$ is desirable. The following theorem shows that this can always be done without decreasing the average work generated in the given process $t(s',w|s)$. 
%Even more, optimal average work can only be achieved when the fluctuations

\begin{theorem}\label{opt}
If $t(s',w|s)$ satisfies condition~\eqref{result 1} then 
\begin{equation}
  \tilde t (s',w|s) =
  \delta (w-w(s'|s)) t(s'|s)
\end{equation}
with
\begin{equation}
  \label{wss}
  w(s'|s) =
  T \ln\! {\int_\mathbb{R} dw\, 
  t(w|s,s')}\, \e^{\beta w}
\end{equation}
also satisfies~\eqref{result 1}, and in addition
\begin{eqnarray}
  \label{average w}
  \langle w\rangle_{\tilde t} 
  &\geq &
  \langle w\rangle_t\ ,
  \\ \label{fluct}
  \max_{w\, :\, \tilde t(w) >0}
  \left| w - \langle w\rangle_{\tilde t} \right| 
  &\leq &
  \max_{w\, :\, t(w) >0}  \left| w - \langle w\rangle_{t} \right| \ .
\end{eqnarray}
The above inequalities are saturated if and only if $\tilde t (w,s'|s) = t(w,s'|s)$.
\end{theorem}

To show that $\tilde t(s',w|s)$ satisfies~\eqref{result 1}, first, exponentiate the two sides of~\eqref{wss}, 
\begin{equation}
  \e^{\beta w(s'|s)} =
  {\int_\mathbb{R} dw\, 
  t(w|s,s')}\, \e^{\beta w}\ ,
\end{equation}
and second, multiply by $t(s'|s)\, \e^{\beta(\mathcal E_{s'} - \mathcal E_\mathrm{S})}$ and sum over $s$,
\begin{equation}
  \sum_s t(s'|s)\, 
  \e^{\beta (\mathcal E_{s'} - \mathcal E_\mathrm{S} + w(s'|s))} =
  \sum_s {\int_\mathbb{R} dw\, 
  t(s',w|s)}\, \e^{\beta (\mathcal E_{s'} - \mathcal E_\mathrm{S} +w)}
  = 1 \ .
\end{equation}
Note that the two maps, $t(s',w|s)$ and $\tilde t(s',w|s)$, have the same value $t(s'|s)$. That is, they perform the same transformation on the system. 

To show~\eqref{average w} we use the convexity of the exponential in equation~\eqref{wss}, obtaining
\begin{equation}
  \label{wss 2}
  w(s'|s) \geq
  T \ln  \e^{\int_\mathbb{R} dw\, 
  P(w|s,s')\, \beta\, w}
  =
  {\int_\mathbb{R} dw\, 
  t(w|s,s')}\, w
  \ .
\end{equation}
Averaging over $s,s'$ gives~\eqref{average w}. Also, note that due to strict convexity of the exponential, the equality in~\eqref{wss 2} is only achieved when $t(w|s,s') = \delta(w-w(s'|s)) = \tilde t (w|s,s')$. 

To show~\eqref{fluct}, note that the convexity of the exponential implies that, unless $t(w|s,s') = \delta(w-w(s'|s))$, there are values $w(+) > w(s'|s)$ and $w(-) < w(s'|s)$ such that $t(w(\pm) |s,s') >0$.
Hence, equality is only achieved when $t=\tilde t$.\\

\section{Supplementary note 2}
\noindent In this supplementary note we calculate the maximal average work extracted (or minimal work of formation) for the state transformation $\rho\rightarrow \rho'$ with unbounded fluctuations. The work is given by the difference in free energy. By explicitly calculating the work distribution we show that the work fluctuations diverge as the initial or final system states become more athermal. We also show that there is no map that can achieve the optimal average work (i.e. thermodynamically reversible) with a smaller range of fluctuations about the average (see also Supplementary Note 1E).\\

In the previous Supplementary Notes we simplified the work-optimal map to a form where it is defined by the map parameters $\{ t(s'|s), w(s'|s)\}$ where $t(s'|s)$ defined the reduced map on the system. We have shown that, in the case that the initial or final state is diagonal we can work with de-phased initial and final states. Furthermore the $t(s'|s)$ must obey 
\begin{eqnarray}
1 &=& \sum\limits_{s'} t(s'|s) \label{x1}\\
x_{s'} &=& \sum\limits_s x_s t(s'|s) \label{x2}\\
1 &=&  \sum_s t(s'|s)\, 
  \e^{\beta (\mathcal E_{s'} - \mathcal E_{s} + w(s'|s))} \label{x3}
\end{eqnarray}
where the first two constraints ensure that the reduced map acting on the system is stochastic and achieves the desired state transformation, and the third constraint, derived in the previous section, is required to ensure that the map is microscopically reversible. In the following we simplify the map further in the case of optimal work extraction. \\

It will be useful to relax the reversibility equality to an inequality
\begin{equation}
1 \geq  \sum_s t(s'|s)\, 
  \e^{\beta (\mathcal E_{s'} - \mathcal E_{s} + w(s'|s))} \label{x4}
\end{equation}
where saturation of the inequality implies that the map is an allowed (thermal) operation. In all future calculations we make use of this relaxed constraint and show that our solutions saturate the inequality. \\

The average work associated with the optimal map achieving $\rho\rightarrow\rho'$ is given by 
\begin{equation}
W=\sum\limits_{ss'}x_s t(s'|s)w(s'|s)
\end{equation}
it turns out to be sufficient to optimize this under the reversibility constraint \eqref{x4}. The Lagrangian is
\begin{equation}
\mathcal{L}= \sum\limits_{ss'}x_s t(s'|s)w(s'|s) + \sum\limits_{s'}\lambda_{s'} \left(1-\sum\limits_{ss'}t(s'|s) e^{\beta (w(s'|s) +\Delta\mathcal{E}_{ss'})} \right)
\end{equation} 
Maximising with respect to $w(s'|s)$ gives 
\begin{equation}
w(s'|s) = \frac{1}{\beta}\log \left(\frac{x_s e^{\beta \Delta\mathcal{E}_{ss'}}}{\lambda_{s'}\beta} \right) \label{wij}
\end{equation}
Extremizing with respect to $\lambda_{s'}$ and applying \eqref{x1} and \eqref{x2} gives 
\begin{equation}
\lambda_{s'} = \frac{x_{s'}}{\beta}
\end{equation}
Substituting this into \eqref{wij} gives the optimal work 
\begin{equation}
\sum\limits_{ss'}x_s t(s'|s) \frac{1}{\beta}\log \left(\frac{x_s e^{\beta \Delta\mathcal{E}_{ss'}}}{x_{s'}} \right) = F(\rho) - F(\rho ')
\end{equation}
where $F(\rho)=\langle \mathcal{E} (\rho) \rangle - 1/\beta\, S(\rho)$ is the standard free energy and we have again used \eqref{x1}, \eqref{x2}. Notice that the result is independent of our choice of $t(s'|s)$, i.e. any map that takes us from $\rho \rightarrow \rho'$ gives the same optimal work. This is simply a statement that the free energy is a state variable (i.e. is  path independent). The work $W^{(\infty )}$ is the average of the work values, gives by 
\begin{equation}
w(s'|s) = \frac{1}{\beta}\log \left(\frac{x_s e^{\beta \Delta\mathcal{E}_{ss'}}}{x_{s'}} \right)\label{wij2}
\end{equation}
Note work values with $x_s =0$ are set to zero (they have zero probability of occurring in the work distribution). Note that the fluctuations diverge as the initial / final state moves further from equilibrium. It is easy to check that substituting our solutions for $w(s'|s)$ into \eqref{x4} saturates the inequality, therefore this map is achievable with thermal operations. \\

Theorem~\ref{opt} shows that no thermal map can achieve this $W$ with smaller worst case fluctuations. Having explicitly calculated the work values, we have derived necessary and sufficient conditions for a thermal map to exist that achieves this $W$ (i.e. the thermodynamically reversible work) given that the fluctuations are constrained $|w-W|\leq c$\\

\noindent \textbf{Result 1.}
The process $(\rho, \, H_\text{S})\rightarrow (\rho' , \, H_\text{S}')$ can be achieved in a thermodynamically reversible map if 
\begin{equation}
e^{\beta (\Delta F(\rho\rightarrow\rho' )-c)}\leq \frac{x_s e^{\beta \mathcal{E}_s}}{x_{s'} e^{\beta \mathcal{E}_{s'}}}\leq  e^{\beta (\Delta F(\rho\rightarrow\rho')+c)} \quad \forall \, s, \, s'
\end{equation}
This becomes a necessary and sufficient condition if the initial and/or or final state is diagonal in the energy eigenbasis
\begin{proof}
Given Theorem \ref{opt}, it is sufficient to find the conditions that the work values \eqref{wij2} obey the $c$-bound
\begin{equation}
\left|\frac{1}{\beta}\log \left(\frac{x_s e^{\beta \Delta\mathcal{E}_{ss'}}}{x_{s'}} \right)-\Delta F (\rho\rightarrow\rho') \right|\leq c
\end{equation}
which gives the desired inequalities
\end{proof}

\section{Supplementary note 3}

In this Supplementary Note we derive the $c$-bounded work content for general quantum state $\rho$.

\subsection{determining the optimal map}

\begin{lemma}\label{simplified map}
We can always find an optimal protocol where $t(s'|s)\rightarrow \tilde{t}_s = e^{-\beta \mathcal{E}_{s'}}/\mathcal{Z}'$, where $\mathcal{Z}'=\text{tr}[e^{-\beta H'_\text{S}}]$, and $w(s'|s) \rightarrow\tilde{w}_s$ that obeys $|\tilde{w}(s)-\langle W \rangle | \leq |w(s'|s)-\langle W \rangle |$ $\forall$ $s,s'$
\end{lemma}
\begin{proof}
The average work extracted by a given protocol is given by 
\begin{equation}
W = \sum\limits_{ss'} x_s t(s'|s) w(s'|s)
\end{equation}
where $\{ t(s'|s), w(s'|s)\}$ obey constraints~\eqref{x1}~\eqref{x2} and ~\eqref{x4}.\\

Assume there exists some optimal protocol $\{ t(s'|s), w(s'|s)\}$ that satisfies these and obeys the $c$-bound
\begin{equation}
|w(s'|s) -W|\leq c \quad \forall s,\, s'
\end{equation}
Define a new protocol with
\begin{align}
\tilde{t}(s') &= \frac{e^{-\beta \mathcal{E}_{s'}}}{\mathcal{Z}'}\\
\tilde{w}(s) &= \sum\limits_k t(k|s) w(k|s) 
\end{align}
The average work extracted by this protocol is 
\begin{equation}
\ \tilde{W} = \sum\limits_{ss'} x_s \tilde{t}(s')  \tilde{w}(s) = \sum\limits_s x_s \tilde{w}(s) = \sum\limits_{ss'} x_s t(s'|s) w(s'|s)
\end{equation}
therefore it extract the same amount of work as the optimal protocol. It is also obeys the microscopic reversibility inequality \eqref{x4} as 
\begin{align*}
\sum\limits_s \tilde{t}(s')  e^{\beta \left(\tilde{w}(s)-\Delta\mathcal{E}_{ss'}\right)} &= \frac{1}{\mathcal{Z}'}\sum\limits_s e^{\beta \left(\sum\limits_k t(k|s) w(k|s) - \mathcal{E}_s \right)} \\
&\leq \frac{1}{\mathcal{Z}'}\sum\limits_{sk}t(k|s) e^{\beta w(k|s)-\beta \mathcal{E}_s}\\
&= \frac{1}{\mathcal{Z}'}\sum\limits_{sk}t(k|s) e^{\beta w(k|s)-\beta \mathcal{E}_s +\beta \mathcal{E}_k-\beta \mathcal{E}_k}\\
&\leq \frac{1}{\mathcal{Z}'}\sum\limits_k e^{-\beta \mathcal{E}_k } \\
&= 1 \\
\end{align*}
where we have used the fact that $\sum_k t(k|s)=1$ and the convexity of the exponential function to get $e^{\beta \tilde{w}(s)} \leq \sum_k t(k|s) e^{\beta w(k|s)} $, and in the second to last line have used the reversibility of the original map. Finally, the new protocol also obeys the $c$-bound as the work values $\tilde{w}(s)$ are convex sums of the work values $w(s'|s)$ and therefore $\text{max}\{w(k|s)\}\geq \tilde{w}(s)\geq \text{min}\{w(k|s)\}$. As $W$ is unchanged then the worst case fluctuations of the new work distribution about the average must be less than or equal to those of the optimal distribution.
\end{proof}

We now drop the tilde from $\tilde{w}(s)$ and $\tilde{t}(s') $. It is simple to check that $t(s') = 1/\mathcal{Z}' e^{-\beta \mathcal{E}_{s'}}$ satisfies the conditions \eqref{x1} and \eqref{x2}. The reversibility inequality \eqref{x4} becomes 
\begin{equation}
\mathcal{Z}'\geq \sum_s e^{\beta (w(s) -\mathcal{E}_s)} \label{reversibility}
\end{equation}
and the average work is given by 
\begin{equation}
W = \sum\limits_s x_s w(s)
\end{equation}

We are now in a position to derive $W^{(c)} (\rho )$. The Lagrangian is the same as employed in the previous section except that it includes terms bounding all $|w(s) - \sum_{k} x_{k} w(k)|\leq 0$. Due to the exponential in the reversibility term $\mathcal{Z}'\geq \sum_s e^{\beta (w(s) - \mathcal{E}_s)}$ we must linearise these bounds so as to avoid generating a transcendental equation when extremizing the Lagrangian over $w(s)$. We therefore select the bounds  
\begin{eqnarray}
w(s) - \sum\limits_{k} x_{k} w(k) &\leq& c \\
\sum\limits_{k} x_{k} w(k)-w(s) &\leq& c
\end{eqnarray}
with associate Lagrange parameters $\mu_{ss'}$ and $\bar \mu_{ss'}$. As all fluctuations are either positive $(w(s) \geq W^{(c)}(\rho))$ or negative $(w(s) < W^{(c)}(\rho ))$ one of these bounds will always be trivial for each work value $w(s)$. Including the reversibility constraint \eqref{reversibility} we maximise the average work $W=\sum_s x_s w(s)$ by extremizing the Lagrangian
\begin{equation}
\mathcal{L}=\sum\limits_s f_s w(s) +\lambda \left( \mathcal{Z}'-\sum\limits_s e^{\beta (w(s) - \mathcal{E}_s)} \right) + c \sum\limits_s \left(\mu_s + \bar{\mu}_s\right)
\end{equation}
where 
\begin{equation}
f_s = x_s\, f-(\bar{\mu}_s-\mu_s) \label{def1}
\end{equation}
where $f=1+\sum_k(\bar{\mu}_k-\mu_k)$ and we have replaced the old Lagrange parameters with $\lambda = \mathcal{Z}'\sum_s \lambda_s $ and $\mu_s (\bar{\mu}_s)= \sum_{s'} \mu_{ss'}(\bar{\mu}_{ss'})$. Maximizing w.r.t $w(i)$ and $\lambda$ gives 
\begin{eqnarray}
w(s) &=& \frac{1}{\beta}\log \left(\frac{f_s e^{\beta \mathcal{E}_s}}{\beta \lambda} \right)\\
\lambda &=& \frac{1}{\beta \mathcal{Z}'}
\end{eqnarray}
where we have used $\sum_s f_s = 1$. Substituting these into the Lagrangian simplifies it to 
\begin{equation}
\mathcal{L}=\frac{1}{\beta}\sum\limits_s f_s \log \left(f_s e^{\beta \mathcal{E}_s}\mathcal{Z}' \right)+c\sum\limits_s (\mu_s +\bar{\mu}_s)
\end{equation}
Maximizing w.r.t $\mu_j$ and $\bar{\mu}_j$ under the condition $\mu_j \geq 0$ and $\bar{\mu}_j\geq 0$ gives 
\begin{eqnarray}
\label{eqA} \\
\frac{\partial \mathcal{L}}{\partial \mu_j}=0 &\rightarrow& \log \left( f_j e^{\beta \mathcal{E}_j}\right)-\sum\limits_s x_s \log \left( f_s e^{\beta \mathcal{E}_s}\right) = -\beta c \nonumber \\
\label{eqB} \\
\frac{\partial \mathcal{L}}{\partial \bar{\mu}_j}=0 &\rightarrow& -\log \left( f_j e^{\beta \mathcal{E}_j}\right)+\sum\limits_s x_s \log \left( f_s e^{\beta \mathcal{E}_s}\right) = -\beta c \nonumber
\end{eqnarray}
where we have used 
\begin{eqnarray}
\frac{\partial f_s}{\partial \mu_j} &=& -x_s + \delta_{sj} \\
\frac{\partial f_s}{\partial \bar{\mu}_j} &=& x_s - \delta_{sj}
\end{eqnarray}
For $c>0$, equations \eqref{eqA} and \eqref{eqB} cannot be simultaneously satisfied by any $w(j)= 1/\beta\,  \log \left(f_j e^{\beta \mathcal{E}_j}\mathcal{Z}' \right)$ therefore if $\mu_j > 0$ then $\bar{\mu}_j=0$ and visa-versa. As $w(i) = 1/\beta\,  \log \left(f_i e^{\beta \mathcal{E}_i}\mathcal{Z}' \right)$ equations \eqref{eqA} and \eqref{eqB} can be written in the form  
\begin{align}
\mu_j \neq 0 \rightarrow w(j) - \langle W \rangle &= - c \label{sat1}\\
\bar{\mu}_j \neq 0 \rightarrow w(j) - \langle W \rangle &= c \label{sat2}
\end{align}
Therefore fluctuations saturate either a positive or negative bound, or saturate no bounds ($\mu_j = \bar{\mu}_j = 0$). It will therefore be useful to partition the set of energy levels into those for which the resulting fluctuations will saturate a positive bound $ i\in \mathcal{X}_+ $, a negative bound $ j\in\mathcal{X}_- $ or saturate no bounds $u\in\mathcal{X}_u$. For work values $w(i)$, $w(i')$ that saturate a positive bound and  $w(j)$, $w(j')$ that saturate a negative bound we have 
\begin{eqnarray}
f_i e^{\beta \mathcal{E}_i } &=& f_{i'} e^{\beta \mathcal{E}_{i'} } \label{sim1}\\
f_j e^{\beta \mathcal{E}_j } &=& f_{j'} e^{\beta \mathcal{E}_{j'} } \label{sim2}\\
f_j e^{\beta \mathcal{E}_j } &=& f_{i'} e^{\beta ( \mathcal{E}_{i'}-2c )} \label{sim3} \\
f_i &=& x_i f - \bar{\mu}_i \label{obv1}\\
f_j &=& x_j f + \mu_j \label{obv2}\\
f_u &=& x_u f \label{obv3}
\end{eqnarray}
where \eqref{sim1} comes from summing \eqref{eqA} for $\mu_j$, $\mu_{j'}$, \eqref{sim2} from summing \eqref{eqB} for $\bar{\mu}_i$, $\bar{\mu}_{i'}$, \eqref{sim3} comes from subtracting \eqref{eqA} for $\mu_j$ and \eqref{eqB} for $\bar{\mu}_i$ and \eqref{obv1}--\eqref{obv3} are just the definition \eqref{def1} with the conditions $\bar{\mu}_j = 0$, $\mu_i = 0$ and $\mu_u = \bar{\mu}_u = 0$ applied. \eqref{sim1}--\eqref{sim3} and \eqref{obv3} allow us to simplify \eqref{eqA} and \eqref{eqB} to 
\begin{eqnarray}
\frac{\partial \mathcal{L}}{\partial \mu_j}&=&0 \rightarrow \log \left(\frac{f}{f_j e^{\beta \mathcal{E}_j}} \right)=\nu +c \label{simp1}\\
\frac{\partial \mathcal{L}}{\partial \bar{\mu}_i}&=&0 \rightarrow \log \left(\frac{f}{f_i e^{\beta \mathcal{E}_i}} \right)=\nu -c \label{simp2}\\
\nu &=& \frac{\beta}{X_u} \left(\frac{1}{\beta}H_u (\rho ) + c(X_- - X_+) - \langle \mathcal{E}_u(\rho ) \rangle \right) \label{v}
\end{eqnarray}

Where $X_+ = \sum_{s\in\mathcal{X}_+} x_s$, $X_- = \sum_{s\in\mathcal{X}_-} x_s$, $X_u = \sum_{s\in\mathcal{X}_u} x_s$, $H_u (\rho ) = - \sum_{s\in\mathcal{X}_u} x_s \log x_s$ and $\langle \mathcal{E}_u(\rho ) \rangle = \sum_{s\in\mathcal{X}_u} x_s \mathcal{E}_s$. \eqref{sim1}--\eqref{obv3} let us to relate the remaining Lagrange parameters by
\begin{eqnarray}
\bar{\mu}_{i'} &=& f (x_{i'} - x_i e^{\beta (\mathcal{E}_i - \mathcal{E}_{i'})} )+\bar{\mu}_i e^{\beta (\mathcal{E}_i - \mathcal{E}_{i'})} \label{rel1} \\
\mu_{j'} &=& f (x_{j}e^{\beta (\mathcal{E}_j - \mathcal{E}_{j'})} - x_{j'} )+\mu_j e^{\beta (\mathcal{E}_j - \mathcal{E}_{j'})} \label{rel2}\\
\label{rel3}\\
\bar{\mu}_i &=& \frac{(k-\mu_j) (e^{\beta \mathcal{E}_i}x_i - e^{\beta (2c - \mathcal{E}_j)}x_j) - \mu_j e^{\beta (2c - \mathcal{E}_j)}}{(1-x_i)e^{\beta \mathcal{E}_i}+x_j e^{\beta (2c-\mathcal{E}_j)}} \nonumber
\end{eqnarray}
where $k= f - \bar{\mu}_i + \mu_j$. Summing \eqref{rel1} over $i'$ and \eqref{rel2} over $j'$ gives 
\begin{eqnarray}
\sum\limits_{i'}\bar{\mu}_{i'} &=& f (X_+ - x_i e^{\beta \mathcal{E}_i}\mathcal{Z}_+)+\bar{\mu}_i e^{\beta \mathcal{E}_i}\mathcal{Z}_+ \\
\sum\limits_{j'}\mu_{j'} &=& f (x_j e^{\beta \mathcal{E}_j}\mathcal{Z}_--X_i)+\mu_j e^{\beta \mathcal{E}_j}\mathcal{Z}_-
\end{eqnarray}
where $\mathcal{Z}_+ = \sum_{s\in\mathcal{X}_+} e^{-\beta \mathcal{E}_s}$ and $\mathcal{Z}_- = \sum_{s\in\mathcal{X}_+}e^{-\beta \mathcal{E}_s}$. $f= 1 + \sum_{i'}\bar{\mu}_{i'} - \sum_{j'}\mu_{j'}$ therefore we can get $f$ in terms of $\bar{\mu}_i$ and $\mu_j$ alone 
\begin{equation}
f=\frac{1+\bar{\mu}_i e^{\beta \mathcal{E}_i}\mathcal{Z}_+-\mu_je^{\beta \mathcal{E}_j}\mathcal{Z}_- }{x_ie^{\beta \mathcal{E}_i}\mathcal{Z}_++x_je^{\beta \mathcal{E}_j}\mathcal{Z}_-+X_u } \label{eqf}
\end{equation}
We now have $f$ in terms of $\bar{\mu}_i$ and $\mu_j$, and \eqref{rel3} relates these to eachother, so we can solve \eqref{simp1} and \eqref{simp2} simultaneously (using $k=f-\bar{\mu}_i + \mu_j$) to find $\bar{\mu}_i$ and $\mu_j$, and by \eqref{rel1}, \eqref{rel2} all Lagrange multipliers. Substituting \eqref{rel1} and \eqref{rel2} into \eqref{sim1} and solving for $\mu_j$ gives 
\begin{equation}
\mu_j = \frac{k e^{\beta \mathcal{E}_i}(e^{-\beta c}-x_j e^{\beta \mathcal{E}_j + \nu })}{e^{\beta (\mathcal{E}_i-c)}+e^{\beta (\mathcal{E}_j + c)} + (1-x_i-x_j)e^{\beta (\mathcal{E}_i + \mathcal{E}_j)+\nu}} \label{mu1}
\end{equation}
and similarly for $\bar{\mu}_i$ 
\begin{equation}
\bar{\mu}_i = \frac{k e^{\beta \mathcal{E}_j }(e^{\beta c}-x_ie^{\beta\mathcal{E}_i + \nu})}{e^{\beta (\mathcal{E}_i -c)}+e^{\beta (\mathcal{E}_j + c)} + (1-x_i-x_j)e^{\beta (\mathcal{E}_i + \mathcal{E}_j)+\nu}} \label{mu2}
\end{equation}
Using \eqref{eqf} and $k=f-\mu_j + \bar{\mu}_i$ simultaneously solve \eqref{mu1} and \eqref{mu2} to give 
\begin{eqnarray}
\bar{\mu}_i &=& \frac{x_i e^\nu - e^{-\beta (\mathcal{E}_i - c)}}{e^\nu X_u + \mathcal{Z}_+e^{\beta c}+\mathcal{Z}_- e^{-\beta c}}\label{mui}\\
\mu_j &=& \frac{ e^{-\beta (\mathcal{E}_j + c)}-x_je^\nu}{e^\nu X_u + \mathcal{Z}_+e^{\beta c}+\mathcal{Z}_- e^{-\beta c}}
\label{muj}
\end{eqnarray}
it is simple to check that \eqref{rel1} and \ref{rel2} result in exactly the same equation for all $\bar{\mu}_{i'}$ and $\mu_{j'}$ but with the corresponding index. Now armed with the explicit form of the Lagrange parameters we have solved the Lagrangian. It is easy to check that for these solutions for $\bar{\mu}_i$, $\mu_j$ the Lagrangian simplifies to $\mathcal{L} = \sum_s x_s w(s)$ where $w(s) = 1/\beta \log (f_s  e^{\beta \mathcal{E}_s }\mathcal{Z}' )$ where the $f_i$ are now of the form 
\begin{equation}
f_s = \gamma^{-1}\begin{cases}
x_s e^\nu \, , \quad s\in \mathcal{X}_u \\
e^{\beta (c-\mathcal{E}_s)}\, , \quad s\in \mathcal{X}_+ \\
e^{-\beta (c+\mathcal{E}_s)}\, , \quad s\in \mathcal{X}_- \\
\end{cases}
\end{equation}  
where 
\begin{equation}
\gamma = e^\nu X_u + \mathcal{Z}_+e^{\beta c}+\mathcal{Z}_- e^{-\beta c}
\end{equation}
Substituting  in the above values for $f_s$ into $\mathcal{L} = 1/\beta \sum_s x_s \log (f_s  e^{\beta \mathcal{E}_s \mathcal{Z}'} )$ gives our final result for work extraction
\begin{equation}
W^{(c)} (\rho ) = \frac{1}{\beta} \left( \log \mathcal{Z}' - \log \gamma \right)\label{mainresult}
\end{equation}
Next we show that $\gamma$ can be written as $\gamma = \sum_s x_s^0 e^{-\beta \mathcal{E}_s}e^{\beta \theta_s}$ there $\theta_S$ is the fluctuation associates with subspace $\ketbra{s}{s}$ which, in the $c$-bounded distribution, is given by 
\begin{equation}
\theta_s = \begin{cases}
\frac{1}{\beta}\log (x_s e^{\beta \mathcal{E}_s})+\frac{\nu}{\beta},\quad  s\in \mathcal{X}_u \\
+c , \quad s\in \mathcal{X}_+ \\
-c ,\quad s\in \mathcal{X}_-
\end{cases}
\end{equation}
In general, the $c$-bounded work is given by the difference between the free energy of the final thermal state and the $c$-bounded free energy
\begin{equation}
F^{(c)}(\rho) = - \log \sum_s x_s^0 e^{-\beta \mathcal{E}_s}e^{\beta \theta_s}
\end{equation}

\subsection{Finding the optimal partition}

\noindent In this Supplementary Note we derive the inequalities for partitioning the state space into positively bounded $\mathcal{X}_+$, negatively bounded $\mathcal{X}_-$ and unbounded $\mathcal{X}_u$ energy levels, as required by our main result. We derive the general set of inequalities for any state and Hamiltonian and give worked through examples of how to find the partition for arbitrary 2 and 3 dimensional systems. Firstly we derive the partition inequalities for work extraction protocols, and then show that the partition inequalities for work of formation are identical. \\

\noindent The Lagrangian is maximised under the condition that $\bar{\mu}_i$ and $\mu_j$ given in \eqref{mu1} and \eqref{mu2} are positive. Therefore if our optimization gives a negative Lagrange parameter it is set to zero, removing the corresponding $c$-bound.  The Lagrange parameters are positive when the following inequalities are satisfied
\begin{eqnarray}
\frac{1}{\beta}\log \left(x_i e^{\beta \mathcal{E}_i} \right) &>&  c - \frac{\nu}{\beta}\quad \,\,\,\,\, \rightarrow \bar{\mu}_i > 0 \label{ineq1}\\
\frac{1}{\beta}\log \left(x_j e^{\beta \mathcal{E}_j} \right) &<& - c - \frac{\nu}{\beta}\quad \rightarrow \mu_j > 0\label{ineq2}
\end{eqnarray}
where 
\begin{equation}
\nu=\frac{\beta}{X_u} \left(\frac{1}{\beta}H_u (\rho ) + c(X_- - X_+) - \langle \mathcal{E}_u(\rho ) \rangle \right)
\end{equation}   
and we have used $\gamma \geq 0$. Clearly $\nu$ depends on how you partition the state space in to positively, negatively and unbounded fluctuations. In the following we derive a set of inequalities that determine the unique partition give $\rho$ and $c$. The following observations simplify the problem\\
\begin{lemma}\label{xulemma}
For any $c$-bounded work distribution where bounded fluctuations saturate their bounds, $X_\pm < 1/2$
\end{lemma}
\begin{proof}
Consider a $c$-bounded work distribution with work values $\{w(s) \}$, average work $W=\sum_s x_s w(s)$, and all work values obey $|w(s) - W|\leq c$. Break the work distribution up into work values that give positive fluctuations $w(i) \geq W $ and negative fluctuations $w(j) < W$. Writing the fluctuations as $\theta_i = w(i) - W$ and $\theta_j = W - w(j)$ and substituting into $W=\sum_s x_s w(s)$ we get 
\begin{equation}
\sum\limits_i x_i \theta_i = \sum\limits_j x_j \theta_j \label{flux}
\end{equation}
which merely states that the average (non-absolute value) fluctuation of a random variable is zero, as is always the case. Take the case that $\sum_{s\in\mathcal{X}_+} x_s > \sum_{s\in\mathcal{X}_-} x_s$ and $X_+ >1/2$. Clearly $\sum_{s\in\mathcal{X}_-} x_s < 1/2$ so $X_-<1/2$. If $X_+\geq 1/2$ then, as bounded fluctuations saturate their bounds, \eqref{flux} gives the inequality 
\begin{equation}
\sum\limits_{s\in\mathcal{X}_-} x_s c < \sum\limits_{s\in\mathcal{X}_+} x_s \theta_s
\end{equation}
taking the factor of $\sum_{s\in\mathcal{X}_+} x_s $ to the other side we can write the inequality 
\begin{equation}
c < \frac{1}{\sum\limits_{s\in\mathcal{X}_-} x_s}\sum\limits_{s\in\mathcal{X}_-} x_s \theta_s
\end{equation}
The right hand side is a convex sum, and all $\theta_s \geq 0$, therefore at least one $\theta_s > c$ contradicting the fact that the work distribution is c-bounded. Therefore we must have $X_+ < 1/2$. A similar argument for the case $\sum_{s\in\mathcal{X}_+} x_s < \sum_{s\in\mathcal{X}_-} x_s$ gives that $X_- < 1/2$.
\end{proof}
The second observation is that the inequalities \eqref{ineq1} and \eqref{ineq2} obey a $\beta$-ordering hierarchy. The $\beta$-odered state, $\rho^{\downarrow\beta}$, is defined as 
\begin{equation}
\rho^{\downarrow\beta} = \left(x_1, x_2, \dots, x_R \right)\quad , \quad x_s e^{\beta \mathcal{E}_s}\geq x_{s+1} e^{\beta \mathcal{E}_{s+1}}
\end{equation}
where $R=\text{rank}(\rho)$ (there is no work value $w^{(\infty )}_s$ associated with $x_s = 0$). The $\beta$-ordering gives $w^{(\infty )}_i \geq w^{(\infty )}_{i+1}$. Therefore if $w^{(\infty )}_i$ satisfies \eqref{ineq1} then so does $w^{(\infty )}_{i-1}$, and if $w^{(\infty )}_j$ satisfies \eqref{ineq2} the so does $w^{(\infty )}_{j+1}$. From this we can deduce that the partition of the state space will look like 
\begin{equation}
\rho^{\downarrow\beta} = (\underbrace{ x_1, x_2, \dots, x_k}_{\mathcal{X}_+},\underbrace{ x_{k+1},\dots, x_{l-1}}_{\mathcal{X}_u},\underbrace{x_l,\dots, x_R}_{\mathcal{X}_-})\label{hirearchy}
\end{equation}

Finally, it will be useful to put the inequalities in to the following form. In the following we drop the $(\infty )$ superscript from the unbounded work values $w^{(\infty )}_s$. Subscript $i$ will label positively bounded fluctuations and subscript $j$ negatively bounded fluctuations. \eqref{ineq1} and \eqref{ineq2} can be written in terms of the unbounded work distribution only (therefore the problem of finding the partition, assuming it exists, is a closed form)
\begin{align}
 W^{(\infty)}(\rho ) &<X_u (w({i'})-c)+\sum\limits_{s\in\mathcal{X}_+} x_s (w(s) -c) \nonumber\\
 &+\sum\limits_{s\in\mathcal{X}_-} x_s (w(s) +c)\quad \forall \, i'\in\mathcal{X}_+ \label{1} \\
 W^{(\infty)}(\rho )&> X_u (w(j')+c)+\sum\limits_{s\in\mathcal{X}_+} x_s (w(s) -c)+\sum\limits_{s\in\mathcal{X}_-} x_s (w(s) +c)\nonumber\\
&\quad \forall \, j'\in \mathcal{X}_- \label{2}
\end{align}
where $W^{(\infty)}(\rho )$ is the unbounded average work $W^{(\infty)}(\rho)=\beta^{-1}\log \mathcal{Z}'+F(\rho)$ and $X_u = 1-X_+ - X_-$. We now prove that for a given $\rho$ and $c$ there exists a unique partition for which all inequalities \eqref{1} and \eqref{2} are satisfied. 

\begin{lemma}\label{partition lemma}
Given state $\rho$, bound value $c$ and inverse temperature $\beta$, there is a unique partition of the state space that gives the optimal $c$-bounded work. For $R=\text{rank}(\rho)$ there are at most $R-1$ inequalities that must be checked to determine the partition.
\end{lemma}

\begin{proof}
By \ref{xulemma} we know that $X_\pm < 1/2$, therefore $X_u > 0$. The partition obeys the $\beta$-ordering hierarchy \eqref{hirearchy}. Starting from $x_1$ count left to right in $\rho^{\downarrow\beta}$ until you find the furthest $x_k$ s.t. $\sum_{i=1}^k x_s < 1/2$. These will form our trial set for the bounded positive fluctuations $\tilde{\mathcal{X}}_+$. Starting from $x_R$ count from right to left until you find the furthest $x_l$ s.t. $\sum_{j=l}^R x_j < 1/2$. These form our trial set for the bounded negative fluctuations $\tilde{\mathcal{X}}_-$. If this is indeed the correct partition for the state space, the tightest bounds will be \eqref{1} on $w(k)$ and \eqref{2} on $w(l)$. The inequalities are 
\begin{align*}
(1-\sum\limits_i x_i - \sum\limits_j x_j) (w(k)-c)&+\sum\limits_i x_i (w(i) -c) \label{3} \numberthis\\
&+\sum\limits_j x_j (w(j) +c) > W(\rho ) \\
(1-\sum\limits_i x_i - \sum\limits_j x_j) (w(l)+c)&+ \sum\limits_i x_i (w(i) -c) \label{4}  \numberthis\\
&+\sum\limits_j x_j (w(j) +c) < W(\rho )
\end{align*}
Whenever an inequality is satisfied it fixes that fluctuation as being the infinitum of its set. For example, if \eqref{3} is satisfied for $w(k)$ then $k\in \tilde{\mathcal{X}}_+$ regardless of $\tilde{\mathcal{X}}_-$. To see this, consider the case that \eqref{3} is satisfied for $w(k)$ but \eqref{4} isn't for $w(l)$ 
\begin{align*}
X_u (w(k_-c) + \sum\limits_{i=1}^k x_i (w(i) -c) + \sum\limits_{j=l}^R x_j (w(j) +c) &> W (\rho ) \\
X_u (w(l)+c) + \sum\limits_{i=1}^k x_i (w(i) -c) + \sum\limits_{j=l}^R x_j (w(j) +c) &> W (\rho )
\end{align*}
as $w(l)$ doesn't satisfy its bound we have to move further down to $w(l'>l)$. This makes $X_u \rightarrow X_u + \sum_{j=l}^{l'} x_j$ and $\sum_{j=l}^R x_j (w(j) +c) \rightarrow \sum_{j=l}^R x_j (w(j) +c) - \sum_{j=l}^{l'} x_j (w(j)+c)$. The left hand side of \eqref{3} gains the term 
\begin{align*}
&+ \sum\limits_{j=l}^{l'} x_j (w(k) - c)-\sum\limits_{j=l}^{l'} x_j (w(j) + c) \numberthis \\
&= \sum\limits_{j=l}^{l'} x_j (w(k)-w(j))
\end{align*}
As $w(k) \geq w(j)$ $\forall$ $j=l, \dots, l'$ this term is positive and the new \eqref{3} is guaranteed to be satisfied. Therefore we see that there is a unique partition as satisfying an inequalities fixes the corresponding (positive or negative) subspace. There are at most $R-1$ inequalities that we need to check, i.e. the ``worst case'' being when the state space is unbounded and we check all $R-1$ inequalities. 
\end{proof}

In summary, the algorithm for determining the partition can be summarised as follows
\begin{itemize}
\item[1.] $\beta$-order the state $\rho^{\downarrow\beta}=(x_1, \dots , x_R)$ where $x_s^{\beta \mathcal{E}_s}\geq x_{s+1}^{\beta \mathcal{E}_{s+1}}$ and $R=\text{rank}(\rho)$. This ordering gives the unbounded work distribution in descending order $(w(1), \dots, x_R)$ where $w(s) = \beta^{-1}\log (x_s e^{\beta \mathcal{E}_s}\mathcal{Z}')$ and $w(s) \geq w(s+1)$
\item[2.] Take the trial partition where you maximise $X_+=\sum_{i=1}^{k} x_i$ under the condition $X_+<1/2$ and $X_-=\sum_{j=l}^{R} x_j$, $X_-<1/2$. Check inequalities 
\begin{align}
X_u (w(k)-c) + \sum\limits_{i=1}^k x_i (w(i) -c) + \sum\limits_{j=l}^R x_j (w(j) +c) &> W (\rho ) \label{5} \\
X_u (w(l)+c) + \sum\limits_{i=1}^k x_i (w(i) -c) + \sum\limits_{j=l}^R x_j (w(j) +c) &< W (\rho ) \label{6}
\end{align}
\item[3.] If \eqref{5} is satisfied $x_k$ fixes $X_+=\sum_{i=1}^k x_i$ and similar for \eqref{6}. Otherwise we perform the next set of inequalities lower in the hierarchy, with $X_+=\sum_{i=1}^{k-1} x_i$ if \eqref{5} is not satisfied and/or $\sum_{j=l+1}^R x_j $ if \eqref{6} is not satisfied. We repeat this process until we find a pair of inequalities that are simultaneously satisfied, fixing $\tilde{\mathcal{X}}_\pm=\mathcal{X}_\pm$, requiring at most $R-1$ inequalities to be checked. 
\end{itemize} 

\noindent \textbf{Case: d=2}. This is the simplest case, as either $x_1> 1/2$ or $x_2 > 1/2$ or they both $=1/2$. In the first case we bound the negative fluctuation, and there is a single bound to check 
\begin{equation}
w(2) < W^{(\infty)}(\rho ) - c
\end{equation} 
and in the second case we bound the positive fluctuation if 
\begin{equation}
w(1) > W^{(\infty)}(\rho ) + c
\end{equation}
and if $x_1=x_2 = 1/2$ when the two fluctuations must be equal (as the average of the positive fluctuations $=$ the average of the negative fluctuations) and we can choose to bound one or the other, giving the same free energy. \\

In the following section we show that the same algorithm is used for determining the partition for the $c$-bounded work of formation. Note that, unless all fluctuations are unbounded $W^{(c)}(\rho) < W^{(\infty )}(\rho)$ and, as shown in the next section, $W^{(c)}_\text{F}(\rho) > W^{(\infty )}_\text{F}(\rho)$.

\section{Supplementary note 4}
In this Supplementary Note we derive the $c$-bounded minimal work of formation ${W}^{(c)}_\text{F}$.

\subsection{determining the optimal work of formation map} 
 Assume there exists some optimal choice of $\{ t(s'|s) , w(s'|s) \}$ that minimize the work cost  whilst obeying the $c$-bounds and microscopic reversibility \eqref{x4}. We start in the thermal state with $x_s = 1/\mathcal{Z}e^{-\beta \mathcal{E}_s}$ and end the protocol in state $\rho$ with probabilities $x_{s'}$. The average work is 
\begin{equation}
W = \sum\limits_{ss'} \frac{1}{\mathcal{Z}}e^{-\beta \mathcal{E}_s}t(s'|s) w(s'|s)
\end{equation}
where the $t_{ij}$ must obey \eqref{x1}, \eqref{x2} and \eqref{x3}.
\begin{lemma}
We can always choose new protocol parameters $\tilde{t}(s')  = x_{s'}$ and $\tilde{w}_{s'} = \sum_{s} \frac{t(s'|s) e^{-\beta\mathcal{E}_s }}{x_{s'} \mathcal{Z}}w(s'|s)$ that give the same average work as the optimal protocol and obey all the necessary constraints
\end{lemma}
\begin{proof}
Clearly these choices of $\tilde{t}(s') $ satisfy $\sum_{s'} \tilde{t}(s')  = 1$ and $\sum_s 1/\mathcal{Z} e^{-\beta \mathcal{E}_s} \tilde{t}(s')  = x_{s'}$, and the average work is given by 
\begin{align*}\numberthis
\tilde{W} &= \sum\limits_{ss} 1/\mathcal{Z} e^{-\beta \mathcal{E}_s} \tilde{t}(s')  \tilde{w}_{s'} = \sum\limits_{s'}  \tilde{t}(s')  \tilde{w}_{s'}\\
&= \sum\limits_{ss'} x_{s'} \frac{t(s'|s)e^{-\beta \mathcal{E}_s} }{\mathcal{Z}x_{s'}}w(s'|s) = W 
\end{align*}
Note that $ \sum_s \frac{t(s'|s) e^{-\beta\mathcal{E}_s }}{x_{s'} \mathcal{Z}} =1$ so $\tilde{w}_{s'}$ is a convex sum of the $w(s'|s)$. The reversibility inequality \eqref{x4} demands that 
\begin{equation}
\sum\limits_s \tilde{t}(s')  e^{\beta (\tilde{w}_{s'} -\Delta\mathcal{E}_{ss'})}\leq 1
\end{equation}
which we can see to be true as the LHS is 
\begin{align*}
\text{LHS}\,\, &= x_{s'} e^{\beta (\tilde{w}_{s'} + \mathcal{E}_{s'})}\sum\limits_s e^{-\beta \mathcal{E}_s }\\ 
&= x_{s'}\mathcal{Z} e^{\beta (\tilde{w}_{s'} + \mathcal{E}_{s'})} \\
&=  x_{s'}\mathcal{Z} e^{\beta\mathcal{E}_{s'}}e^{\beta \sum_i \frac{t(s'|s) e^{-\beta\mathcal{E}_s }}{x_{s'} \mathcal{Z}}w(s'|s)}\\
&\leq x_{s'}\mathcal{Z} e^{\beta\mathcal{E}_{s'}}\sum\limits_s \frac{t(s'|s) e^{-\beta\mathcal{E}_s }}{x_{s'} \mathcal{Z}}e^{\beta w(s'|s)}\\
&= \sum\limits_s t(s'|s) e^{\beta (w(s'|s) - \Delta \mathcal{E}_{ss'})} \leq 1
\end{align*}
where in the fourth step we have used the convexity of the exponential function and in the last step we have used the fact that the optimal protocol defined by $\{ t(s'|s) , w(s'|s)\}$ is reversible.\\
Clearly the new protocol obeys the $c$-bounds as $\tilde{w}_{s'}$ is a convex combination of $w(s'|s)$ and the average of the work distribution is the same for both maps, therefore the spread of $\tilde{w}_{s'}$ about the average is less than or equal to that of $w(s'|s)$.
\end{proof}

Substituting $t(s'|s)\rightarrow\tilde{t}(s')  = x_{s'}$ and $w(s'|s)\rightarrow\tilde{w}_{s'}$ into the Lagrangian and simplifying gives 
\begin{align*}
\mathcal{L}&=\sum\limits_{s'} f_{s'} w(s') + \sum\limits_{s'} \lambda_{s'} \left(1-x_{s'} e^{\beta \mathcal{E}_{s'}}\mathcal{Z}e^{\beta w(s')} \right)\\ 
&+ c \sum\limits_{s'} (\mu_{s'} + \bar{\mu}_{s'} )
\end{align*}

where we have dropped the tilde from $\tilde{w}_{s'}$ and 
$f_{s'} = x_{s'}(1+f) - (\bar{\mu}_{s'} - \mu_{s'})$ and $f=\sum\limits_{s'}( \bar{\mu}_{s'}-\mu_{s'})$. $\bar{\mu}$  give the bounds on positive fluctuations and $\mu$ give the bounds for negative fluctuations. Maximizing w.r.t $w(s')$ and $\lambda_{s'}$ gives
\begin{eqnarray}
\lambda_{s'} &=&  \text{max}\{\, \frac{f_{s'}}{\beta}\, , \, 0\,  \} \\
w(s') &=& -\frac{1}{\beta}\log \left(\frac{\beta \lambda_{s'} x_{s'} e^{\beta \mathcal{E}_{s'}}\mathcal{Z} }{f_{s'}} \right)\\
&=& -\frac{1}{\beta}\log \left( x_{s'} e^{\beta \mathcal{E}_{s'}}\mathcal{Z} \right)\,\, , \quad f_{s'} \neq 0
\end{eqnarray}
which is the work value given in the unbounded cases (it contains no $\bar{\mu}_j, \mu_j$ terms). Therefore all fluctuations with $f_{s'} >0$ are unbounded, with $\mu_{s'} (\bar{\mu}_{s'})=0$, giving $f_{s'}=x_{s'}$ if unbounded or $0$ if bounded. The Lagrangian reduces to 
\begin{equation}
\mathcal{L}=-\frac{1}{\beta} \sum\limits_{s'} f_{s'} \log \left( x_{s'} e^{\beta \mathcal{E}_{s'}}\mathcal{Z} \right) + c \sum\limits_{s'} \left(\bar{\mu}_{s'} + \mu_{s'} \right)
\end{equation}
Which we can now re-write in terms of bounded and unbounded fluctuations. Index all fluctuations with $f_{s'} > 0$ with $s'\in\mathcal{X}_u$, and all those with $f_{s'} = 0$ with $s'\in\mathcal{X}_+$. The Lagrangian becomes
\begin{equation}
\mathcal{L}=-\frac{1}{\beta} \sum\limits_{s'\in\mathcal{X}_u} x_{s'} (1+f) \log \left( x_{s'} e^{\beta \mathcal{E}_{s'}}\mathcal{Z} \right) + c \sum\limits_{s'\in\mathcal{X}_+} \left(\bar{\mu}_s + \mu_s \right)
\end{equation}
$f_s=0$ $\forall s'\in\mathcal{X}_+$ gives
\begin{equation}
\bar{\mu}_{s'} - \mu_{s'} = x_{s'} (1+f)
\end{equation}
Summing over $s\in\mathcal{X}_+$ and solving for $f=\sum\limits_{k'} (\bar{\mu}_{k'} - \mu_{k'})$ gives 
\begin{equation}
f = \frac{X_c}{1-X_c}
\end{equation} 
where $X_c = \sum_{s'\in\mathcal{X}_+} x_s$. For a given bounded fluctuation only one of the $\mu$ Lagrange parameters is non-zero, corresponding to if the fluctuation is positive or negative. Using the above result gives
\begin{eqnarray}
\bar{\mu}_k &=& x_k (1+\frac{X_c}{1-X_c})= \frac{x_k}{1-X_c} \label{aaa}\\
\mu_l &=& - \frac{x_l}{1-X_c}
\end{eqnarray}
The Lagrange parameters are restricted to being positive, $\mu_s \geq 0$ and $\bar{\mu}_s \geq 0$ $\forall$ $s$. Therefore as $\mu_s = - x_s / (1-X_c)$ which is $\leq 0$, therefore all $\mu_s = 0$ and we never bound negative fluctuations. \\

Using $X_u = 1-X_c$ we arrive at the $c$-bounded optimal work of formation
\begin{equation}
W_F^{(c)} = \frac{1}{X_u}\left( -\frac{1}{\beta} \sum\limits_{s'\in\mathcal{X}_u} x_{s'}\log \left( x_{s'} e^{\beta \mathcal{E}_{s'}}\mathcal{Z} \right)+ c \, X_c\right) \label{workofformation}
\end{equation}
This is exactly the work cost when we calculate the $c$-bounded work of formation in the following way. Take the optimal unbounded work distribution for forming a state $(w(1) , \dots , w(d))$ with average $W^{(\infty )} (\rho )$. Equation \eqref{workofformation} is given by the solution to the equation 
\begin{align*}
W &= \sum\limits_{s\in\mathcal{X}_u} x_s w(s) +\sum\limits_{s\in\mathcal{X}_+} x_s (W+c)\\
\therefore W &\underbrace{(1-\sum\limits_{s\in\mathcal{X}_+} x_s)}_{X_u} = \sum\limits_{s\in\mathcal{X}_u} x_s w(s) + c X_+ \\
\therefore W &= \frac{1}{X_u}\left( \sum\limits_{s\in\mathcal{X}_u} x_s w(s) + c X_+\right)
\end{align*}
I.e we simply take the optimal unbounded work distribution and if the largest positive fluctuation breaks $w(i)-W\leq c$ we replace it with $W'+c$, recalculating the average each time. 

\subsection{finding the optimal partition for state formation}

We now show this algorithm for finding the partition is identical to the algorithm derived in Supplementary Note 2 for work extraction, in the case of positive fluctuation bounds only. \\

\begin{theorem}\label{partitionformation}
The partition of the state space that gives the $c$-bounded work of formation is found using the algorithm described in Lemma \ref{partition lemma} but bounding only positive fluctuations. 
\end{theorem}

\begin{proof}
Take the final state $\rho'$ and $\beta$-order it
\begin{equation}\rho'^{\downarrow\beta}=(x_1 ,\dots x_R) \quad , \quad x_{s'}e^{\beta \mathcal{E}_{s'}}\geq x_{s'+1}e^{\beta \mathcal{E}_{s'+1}}
\end{equation}
where $R=\text{rank}(\rho)$ (we can discount work values in the unbounded work distribution associated with $x_s=0$ as they have no probability of being observed). The unbounded work distribution obeys the inverse $\beta$-ordering as $w(s') = -\beta^{-1}\log (x_{s'} e^{\beta\mathcal{E}_{s'}}\mathcal{Z})$, therefore $w(1)\leq w(2) \leq \dots \leq w(R)$. For a correct partition where $X_+ = \sum_{s'=k}^R x_i$ we require that $w(s') \leq W_F (\rho')^{(c)} + c$ $\forall$ $s<k$. As with the $c$-bounded work extraction protocol the $\beta$-ordering puts these inequalities into a hierarchy, with the bound for $w(k-1)$ being the tightest. As all bounds must be satisfied we need only check the tightest one. For the aforementioned partition the tightest bound is 
\begin{equation}
w(k-1) \leq \frac{1}{\sum\limits_{i=k}^R x_i}\left(W^{(\infty )} - \sum\limits_{s'=k}^R x_{s'} w(s') + c \sum\limits_{s'=k}^R x_{s'} \right)
\end{equation}
which can be re-arranged to give 
\begin{equation}
(1-\sum\limits_{s'=k}^R x_{s'}) w(k-1) + \sum\limits_{{s'}=k}^R x_{s'} w(s') \leq W^{(\infty )} + c
\end{equation}
which is simply the inequality used for checking a partition of the state space in the $c$-bounded work extraction protocol, in the case that there are no negative fluctuations that saturate their bounds \eqref{5}. It can be simplified further to 
\begin{equation}
\sum\limits_{s'=1}^{k-1}x_{s'} (w(k-1)-w(s') ) \leq c \label{counter}
\end{equation}
So the partition is defined by the largest $k$ s.t. 
\begin{align}
\sum\limits_{s'=1}^{k}x_{s'} (w(k)-w(s') ) &> c \label{formation inequality}\\
\sum\limits_{s'=1}^{k-1}x_{s'} (w(k-1)-w(s') ) &\leq c
\end{align}
All $w(s)$ with $s\geq k$ are bounded and all with $s<k$ are unbounded. Once again we have at most $R-1$ inequalities to check. I.e. starting from $k=2$ we check \eqref{formation inequality} for $k=2,3,\dots $. 
\end{proof}

\section{Supplementary note 5}

\subsection{recovering standard and single-shot regimes}
\begin{theorem}
$\lim\limits_{c\rightarrow 0}W^{(c)} (\rho ) = W_\text{min} (\rho ) $
\end{theorem}

\begin{proof}

As $c\rightarrow 0$ the inequalities that determine if a fluctuation saturates its bound \eqref{ineq1} and \eqref{ineq2} become 
\begin{align*}
\frac{1}{\beta}\log \left(x_s e^{\beta \mathcal{E}_s} \right) &\geq -\nu / \beta \rightarrow s \in \mathcal{X}_+ \\
\frac{1}{\beta}\log \left(x_s e^{\beta \mathcal{E}_s} \right) &\leq -\nu / \beta \rightarrow s \in \mathcal{X}_- \\
\end{align*}
All fluctuations satisfy one of these bounds, except in the case that $x_s = 0$ (i.e. there is no fluctuation associated with state $\ket{s}$), therefore $X_u \rightarrow 0$. $\gamma = X_u e^\nu + \mathcal{Z}_+e^{\beta c}+\mathcal{Z}_-e^{-\beta c} \rightarrow \mathcal{Z}_+ + \mathcal{Z}_-$ which can be formulated as 
\begin{equation}
\gamma = \sum\limits_s x_s^0e^{-\beta \mathcal{E}_s}
\end{equation} 
therefore $W^{(c)} (\rho )$ becomes 
\begin{align*}
\lim\limits_{c\rightarrow 0}W^{(c)} (\rho ) &= \frac{1}{\beta}\left(\log \mathcal{Z}' - \log \sum\limits_s x_s^0 e^{-\beta \mathcal{E}_s}  \right) \numberthis\\
&= W_\text{min} (\rho )
\end{align*}
\end{proof}

\begin{theorem}
$\lim\limits_{c\rightarrow 0}W_F^{(c)} (\rho ) = W^\text{min}_F (\rho ) $
\end{theorem}

\begin{proof}
Following the partition algorithm derived in Lemma \ref{partitionformation}, clearly when $c\rightarrow 0$ we bound all but the most negative fluctuation(s), those with value $w(1)$, as this gives the only positively bounded work distribution for which all work values are $\leq $ $W^{(0)}_F$. Any other choice of partition would require both negative and positive bound saturation, which contradicts \eqref{aaa} being positive. Therefore $W^{(0)}_F$ is given by 
\begin{equation}
W^{(0)}_F=\min\limits_s \left\{-\frac{1}{\beta}\log\left(x_s e^{\beta\mathcal{E}'_s}\mathcal{Z}\right) \right\} = -\frac{1}{\beta}\log\left(x_1 e^{\beta\mathcal{E}'_1}\mathcal{Z}\right)
\end{equation}
where $\rho'^{\downarrow\beta}=(x_1,\dots,x_R)$. We can interpret 
\begin{equation}
w(1) = -\frac{1}{\beta}\log\left(x_1 e^{\beta\mathcal{E}'_1}\mathcal{Z}\right)
\end{equation}
$-w(1)$ is the ``on ramp'' (the first segment) of the Lorenz curve of state $\rho'$ (see Figure~\ref{lfig} and \citep{horodecki2013fundamental}). It therefore gives the work that can be extracted from the ``thermally sharp state'', where all segments either have the same gradient or zero gradient, that just thermomajorizes $\rho'$, see figure \ref{lfig} below and also \citep{horodecki2013fundamental}. This can also be interpreted as the upper bound to the work that can be in-deterministically extracted from $\rho'$, and this is precisely the single-shot work of formation.
\begin{figure}[h!]
\includegraphics[scale=0.5]{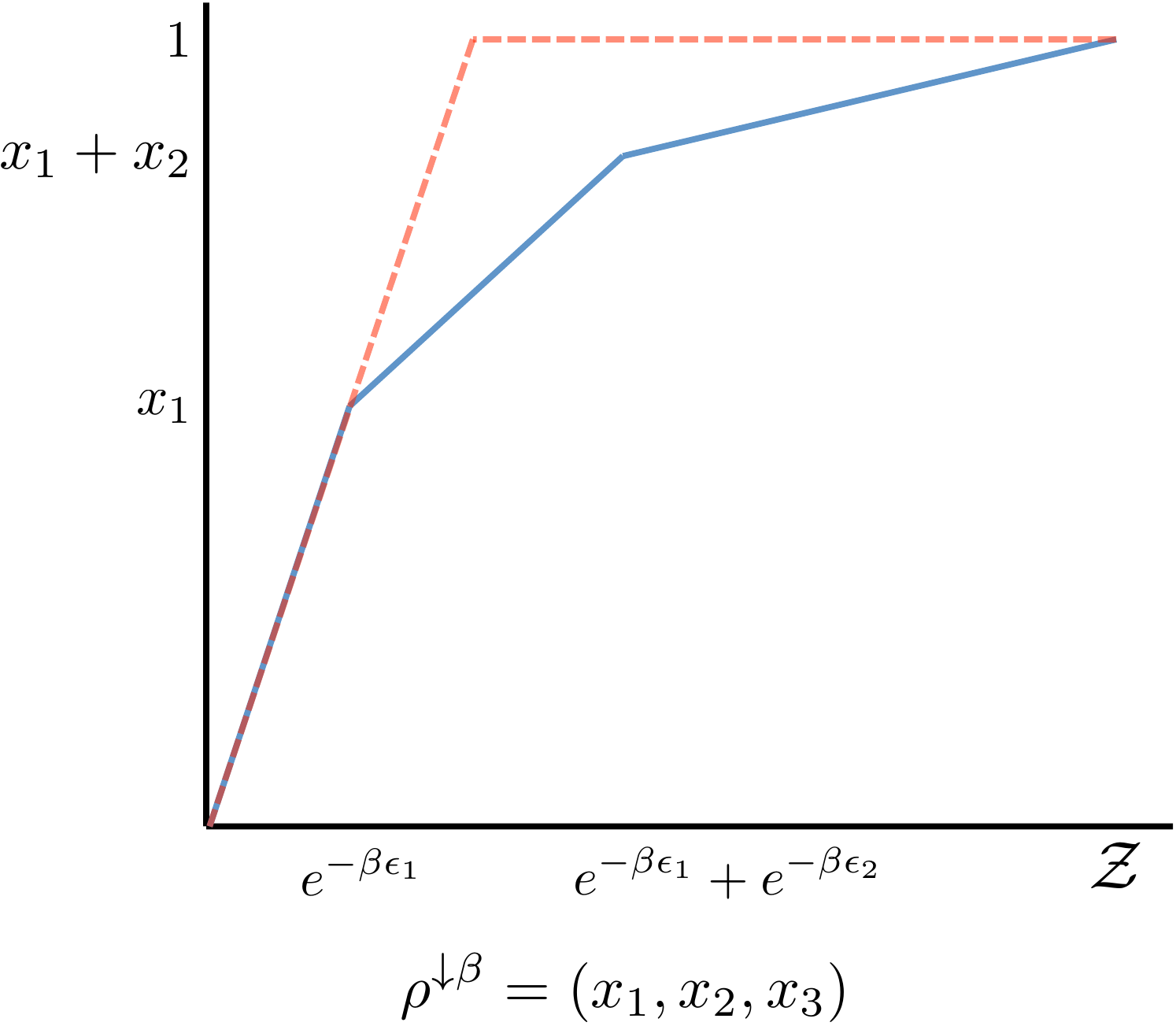}
\caption{Lorenz curve depicting work of formation of a state. Figure shows the thermally sharp state whose single-shot work content gives the work of formation of state $\rho'$. This also represents the maximal average work that can be extracted from $\rho'$ }\label{lfig}
\end{figure}
\end{proof}

\subsection{bounds on c-bounded work}

\begin{lemma}
$W^{(c)} (\rho )  \leq W_\text{min}(\rho ) + c$. $c\in [ 0, c_\text{crit}]$, therefore any gain in work extracted over the single-shot work context of $\rho$ requires fluctuations that are at least as large as the increase in extracted work
\end{lemma} 
\begin{proof}
For $c=0$ we get that $W_{c=0}(\rho ) = W_\text{min} (\rho)$. $W^{(c)}$ increases as we increase $c$ from 0. We show $W^{(c)}(\rho)$ grows sub-linearly with $c$, i.e. that
\begin{equation}
\frac{\partial W^{(c)}(\rho)}{\partial c} \leq 1 \quad \forall \, c
\end{equation}
which implies that $|W^{(c)} (\rho ) - W_\text{min}(\rho )| \leq c$. Using \eqref{mainresult} the derivative of $W^{(c)}(\rho )$ with respect to $c$ is 
\begin{equation}
\frac{\partial W^{(c)}(\rho)}{\partial c} = -\frac{1}{\beta \gamma}\frac{\partial \gamma}{\partial c} =  \frac{1}{\gamma} \left((X_+ - X_-)e^\nu - \mathcal{Z}_+e^{\beta c}+\mathcal{Z}_-e^{-\beta c} \right)
\end{equation}
where $\gamma = X_u e^\nu + \mathcal{Z}_+e^{\beta c}+\mathcal{Z}_-e^{-\beta c}$. For this to be greater than $1$ we require that 
\begin{equation}
(X_+ - X_-)e^\nu - \mathcal{Z}_+e^{\beta c}+\mathcal{Z}_-e^{-\beta c} > X_u e^\nu + \mathcal{Z}_+e^{\beta c}+\mathcal{Z}_-e^{-\beta c}
\end{equation}
As $\mathcal{Z}_\pm e^{\pm \beta c} \geq 0$ and $e^\nu \geq 0$ this cannot be satisfied unless 
\begin{equation}
X_+ - X_- > X_u \xrightarrow{\therefore} X_+ > \frac{1}{2}
\end{equation}
where we have used $X_u = 1-X_+ - X_-$. We now show that this is never true for $c$-bounded work distributions. Using $W^{(c)}(\rho) = \sum_s x_s \theta_s$, where $\theta_s = w(s) - W^{(c)}(\rho )$ is the fluctuation associated with work value $w(s)$, gives $\sum_s x_s \theta_s = 0$. Divide the fluctuations for the $c$-bounded protocol $\theta_s$ into positive $\theta_\alpha \geq 0$ where $\alpha\in\mathcal{X}_+$ and negative $\theta_\beta <0$ where $\beta\in\mathcal{X}_-$, giving
\begin{equation}
\sum\limits_\alpha x_\alpha |\theta_\alpha | = \sum\limits_\beta x_\beta |\theta_\beta | \label{fluc2}
\end{equation}
according to \eqref{sat1} and \eqref{sat2}, bounded fluctuations saturate their bounds. Take the subset of the positive fluctuations that saturate their bounds as $\theta_{\alpha'}$. \eqref{fluc2} becomes
\begin{equation}
c X_+ + \sum\limits_{\alpha\neq \alpha '}x_\alpha |\theta_\alpha | = \sum\limits_\beta x_\beta |\theta_\beta |
\end{equation}
Dividing both sides by $\sum_\beta x_\beta$ gives 
\begin{equation}
\frac{\sum\limits_\beta x_\beta |\theta_\beta |}{\sum\limits_\beta x_\beta }= \frac{c X_+ + \sum\limits_{\alpha\neq \alpha '}x_\alpha |\theta_\alpha |}{\sum\limits_\beta x_\beta}
\end{equation}
the right hand side is strictly larger than $c$ as $X_+ > 1/2$ and therefore  $\sum_\beta x_\beta < 1/2$, and $\sum_{\alpha\neq \alpha '}x_\alpha |\theta_\alpha | \geq 0$. The left hand side is a convex sum with all $|\theta_\beta| \geq 0$, therefore at least one $|\theta_\beta |$ must be larger than $c$. Therefore in a $c$-bounded protocol $X_+\leq 1/2$.
\end{proof}

\begin{lemma}
$\left| W^{(c)}_F (\rho' )\right|  \geq \left| W_\text{max}(\rho' )\right| - c$. $c\in [ 0, c_\text{crit}]$, therefore any gain in work extracted over the single-shot work of formation of $\rho'$ requires fluctuations that are at least as large as the increase in extracted work
\end{lemma}

\begin{proof}
if $\left| W^{(c)}_F (\rho' )\right|  < \left| W_\text{max}(\rho' )\right| - c$ then, taking $w(1) = \text{max}\{ w(s') \}=\text{max}\{ \beta^{-1}\log x_{s'} e^{\beta \epsilon_{s'}}\mathcal{Z}\}=\left| W_\text{max}(\rho' )\right|$ when we require that 
\begin{equation}
\frac{1}{X_u}\left(\sum\limits_{s'\in\mathcal{X}_u}w(s')+c(1-X_u) \right)+c < w(1)
\end{equation}
which requires that
\begin{equation}
c< \sum\limits_{s'\in\mathcal{X}_u}(w(1)-w(s'))
\end{equation}
which contradicts~\eqref{counter} for any partition
\end{proof} 

\subsection{c-bounded work obeys the Jarzinski equality}
In this Supplementary Note we show that all $c$-bounded work distributions obey the Jarzinski equality. The equality is given in the case that a system begins in the thermal state of an initial Hamiltonian $H_i$ and the Hamiltonian is transformed to $H_f$, causing the state to evolve to a state that is, potentially, out of equilibrium with the final Hamiltonian. The equality is stated as\\
\begin{equation}
\langle e^{-\beta w}\rangle = e^{-\beta \Delta F}
\end{equation}
where $w$ is the work random variable (i.e. the work values associated with the protocol), with the convention that they are positive if the work is done on the system, and $\Delta F$ is the free energy difference between the equilibrium states of $H_\mathcal{S}$ and $H'_\mathcal{S}$, given by \\
\begin{equation}
\Delta F = \frac{1}{\beta}\log\left(\frac{\mathcal{Z}}{\mathcal{Z}'} \right)
\end{equation}

\begin{theorem}\label{jarzinskilemma}
Any protocol that saturates all reversibility constraints immediately satisfies the Jarzinski equality
\end{theorem}
\begin{proof}
If a protocol, characterised by $\{t(s'|s), w(s'|s) \}$, saturates all reversibility constraints then 
\begin{equation}
\sum\limits_s t(s'|s)e^{\beta (w(s'|s) - \mathcal{E}_s + \mathcal{E}_{s'})} = 1
\end{equation}
Dividing by $\mathcal{Z}$ and re-arranging gives
\begin{equation}
\sum\limits_s \frac{e^{-\beta\mathcal{E}_s}}{\mathcal{Z}}t(s'|s)e^{\beta w(s'|s)}=\frac{e^{-\beta \mathcal{E}_{s'}}}{\mathcal{Z}}
\end{equation}
if we sum over $s'$, the LHS of this is equivalent to $\langle e^{\beta w(s'|s)}\rangle$ with the initial state being thermal w.r.t $H_\mathcal{S}$, $x_s=\mathcal{Z}^{-1}e^{-\beta\mathcal{E}_s}$, and the RHS becomes $\mathcal{Z}'/\mathcal{Z}$, which is equal to $e^{\beta \Delta F}$. Note that we use the convention of positive work for work extracted from the system, hense our $w(s'|s)=-w(s)$ are the negatives of the work values given in the Jarzinski equality. \\
\end{proof}

\begin{theorem}
The work distributions associated with $W^{(c)}(\rho)$ and $W_F^{(c)}(\rho)$ obey the Jarzinski equality. 
\end{theorem}
First we cover the case of work extraction, $W^{(c)}(\rho)$. The work values are give by 
\begin{equation}
w(s) = \frac{1}{\beta}\log (f_s e^{\beta\mathcal{E}_s}\mathcal{Z}') 
\end{equation}
and the $t(s'|s)$ are given by 
\begin{equation}
t(s') = \frac{e^{-\beta\mathcal{E}'_{s'}}}{\mathcal{Z}'}
\end{equation}
as derived in Supplementary Note 2. Note that 
\begin{equation}
\sum\limits_s f_s = 1
\end{equation}
The LHS of the reversibility constraints are 
\begin{align*}
&= \sum\limits_s \frac{e^{-\beta\mathcal{E}_{s'}}}{\mathcal{Z}'} e^{\beta (w(s) - \mathcal{E}_s + \mathcal{E}_{s'})} \\
&= \sum\limits_s \frac{1}{\mathcal{Z}'} e^{\log (f_s \mathcal{Z}'))} \\
&= \sum\limits_s f_s = 1
\end{align*}
and can similarly be shown for $W^{(c)}(\rho)$.

\section{Supplementary note 6}

In this Supplementary Note we calculate the $c$-bounded Carnot efficiency for a single qubit quantum engine. The engine operates by moving a qubit $\rho = x \ketbra{0}{0}+(1-x)\ketbra{1}{1}$ with Hamiltonian $H_s=\mathcal{E}\ketbra{0}{0}$ between two baths with inverse temperature $\beta_H$ and $\beta_C$, with $\beta_H<\beta_C$. The engine cycle begins with the qubit in thermal equilibrium with the cold bath. It is then placed in contact with the hot bath, and extracting work by allowing the state to equilibrate. The final step in the cycle is to return the qubit to the cold bath, allowing it to equilibrate and extracting work.\\ 
\textbf{Case: }$\mathbf{c\rightarrow\infty}$. Following the proof given in \cite{skrzypczyk2014work}, we show that in the case that fluctuations are unbounded it is possible to reach Carnot efficiency. The equilibrium state of the qubit if given by \\
\begin{equation}
\rho_{H,C}=\frac{1}{Z_{H,C}}e^{-\beta_{H,C}\mathcal{E}}\ketbra{0}{0} +\frac{1}{Z_{H,C}}\ketbra{1}{1}
\end{equation}
where $\mathcal{Z}_{H,C}=e^{-\beta_{H,C}\mathcal{E}}+1$. When $\rho_C$ is placed in thermal contact with the hot bath we can extract a maximal average work given by the free energy difference 
\begin{align*}
F(\rho_C,\beta_H)-F(\rho_H,\beta_H)&=\frac{1}{\beta_H}\log\left(\frac{\mathcal{Z}_H}{\mathcal{Z}_C} \right)\\
&+\text{tr}\left( H_s \rho_c \right) \left(\frac{\beta_H-\beta_C}{\beta_H} \right) \numberthis
\end{align*}
where $\text{tr}\left( H_s \rho_C \right)=\mathcal{Z}^{-1}_Ce^{-\beta_C\mathcal{E}}\mathcal{E}$ and work values\\
\begin{align*}
w_{0,H} &= \frac{1}{\beta_H}\log\left(\frac{\mathcal{Z}_H}{\mathcal{Z}_C}e^{\mathcal{E} (\beta_H - \beta_C)} \right) \\
w_{1,H} &= \frac{1}{\beta_H}\log\left(\frac{\mathcal{Z}_H}{\mathcal{Z}_C}\right) \numberthis
\end{align*}
where we have used the result for the optimal unbounded work extraction protocol $w(s) = \beta^{-1}\log x_s e^{\beta\mathcal{E}_s}\mathcal{Z}$. Returning the qubit to the cold bath we extract work
\begin{align*}
F(\rho_H,\beta_C)-F(\rho_C,\beta_C)&=\frac{1}{\beta_C}\log\left(\frac{\mathcal{Z}_C}{\mathcal{Z}_H} \right)\\
&+\text{tr}\left( H_s \rho_H \right) \left(\frac{\beta_C-\beta_H}{\beta_C} \right) \numberthis
\end{align*}
with work values \\
\begin{align*}
w_{0,C} &= \frac{1}{\beta_C}\log\left(\frac{\mathcal{Z}_C}{\mathcal{Z}_H}e^{\mathcal{E} (\beta_C - \beta_H)} \right) \\
w_{1,C} &= \frac{1}{\beta_C}\log\left(\frac{\mathcal{Z}_C}{\mathcal{Z}_H}\right) \numberthis
\end{align*}
where $\text{tr}\left( H_s \rho_H \right)=\mathcal{Z}^{-1}_He^{-\beta_H\mathcal{E}}\mathcal{E}$. The total work extracted in one cycle is 
\begin{equation}
W_\text{tot}=\left(\frac{1}{\beta_H}- \frac{1}{\beta_C}\right)\left(S_H-S_C \right)
\end{equation}
where $S_{H,C} = -\log \mathcal{Z}_{H,C}-\beta_{H,C}\text{tr}\left( H_s \rho_{H,C} \right)$ is the Von Neumann entropy of the Gibbs state with inverse temperature $\beta_{H,C}$. Applying the first law of thermodynamics $\Delta U = Q-W$ to the first step (extracting work from the hot bath), we get 
\begin{equation}
\text{tr}\left( H_s \rho_{H} \right)-\text{tr}\left( H_s \rho_{C} \right)=Q_H - \left(F(\rho_C,\beta_H)-F(\rho_H,\beta_H) \right)
\end{equation}
which simplifies to 
\begin{equation}
Q_H = \frac{1}{\beta_H}\left(S_H - S_C \right)
\end{equation}
where $Q_H$ is the heat flow out of the hot bath. The efficiency is given by the ratio of the total work to $Q_H$, giving 
\begin{equation}
\frac{W_\text{tot}}{Q_H}=1-\frac{\beta_H}{\beta_C}
\end{equation}
which is the Carnot efficiency.\\
\textbf{Case:} $\mathbf{c}$ \textbf{finite}. \\
$\mathcal{E}>0$ implies that $x<1/2$ for all thermal states of $\rho$, therefore by the partitioning algorithm derived in Supplementary Note 4, the least likely fluctuation $w^{(c)}_0$ associated with subspace $\ket{0}$ is the only fluctuation that can be bounded. Using the result of Supplementary note 3, the $c$-bounded work content for a qubit $\rho^{\downarrow\beta} = (x_1, x_2)$ is given by \\
\begin{equation}
W^{(c)}(\rho) = \frac{1}{\beta}\log\mathcal{Z}+\begin{cases}
F(\rho ) \, , \,\quad  x_1 e^{\beta\mathcal{E}_1}\leq e^{\beta (F(\rho)+c)}\\
\quad\quad \quad  \quad x_2 e^{\beta\mathcal{E}_2}\geq e^{\beta (F(\rho)-c)}\\
F^{(c)}_1(\rho) \, , \,  w(0)< W(\rho) -c\\
F^{(c)}_2(\rho) \, , \,    w(0)> W(\rho) +c\\
\end{cases}\label{cboundedwork qubit}
\end{equation}
where $F^{(c)}_1(\rho)=- \beta^{-1}\log\left(e^{-c \beta+ \frac{c\beta}{1-x}}+e^{-c\beta - \mathcal{E}\beta}\right)$ and $F^{(c)}_2(\rho)=-\beta^{-1}\log\left(e^{c \beta- \frac{c\beta}{1-x}}+e^{c\beta - \mathcal{E}\beta}\right)$, $W(\rho)$ is the unbounded work given by the free energy, and $w(0) = \beta^{-1}\log (x e^{\beta \mathcal{E}}\mathcal{Z})$. In the first step, putting the qubit in contact with the hot bath, the work we extract is bounded negatively if $w_{0,H}<F(\rho_C,\beta_H)-F(\rho_H,\beta_H)-c$, which gives the inequality 
\begin{equation}
\frac{\mathcal{E}}{\mathcal{Z}_C}\left(\frac{\beta_C-\beta_H}{\beta_H} \right)>c \label{carnot1}
\end{equation}
note that $\beta_C > \beta_H$ therefore the LHS of \eqref{carnot1} is positive. In order to break the positive fluctuation bound we would require -LHS $>c$ which is impossible for positive $c$, therefore in the first part of the engine cycle the work extracted is either unbounded or negatively bounded. For the second cycle, re-equilibriating the qubit with the cold bath, the inequality that implies a positively bounded work is given by $w_{0,C}>F(\rho_H,\beta_C)-F(\rho_C,\beta_C)+c$, which simplifies to
\begin{equation}
\frac{\mathcal{E}}{\mathcal{Z}_H}\left(\frac{\beta_C-\beta_H}{\beta_C} \right)>c \label{carnot2}
\end{equation}
Again the LHS is positive, and we would require -LHS $>c$ in the case of a negatively bounded protocol. Therefore on the second part of the cycle the work is either unbounded or positively bounded. Also note that as $\beta_H<\beta_C$ and therefore $\mathcal{Z}_C<\mathcal{Z}_H$, satisfying inequality \eqref{carnot1} implies inequality \eqref{carnot2} is also satisfied. There are therefore three cases, 1) the work is unbounded and we can achieve Carnot efficiency, 2) the work extracted from the hot bath is negatively bounded, and 3) 2 and the work extracted from the cold bath is positively bounded. \\

the $c$-bounded work extracted from the hot bath, in the case that inequality \eqref{carnot1} is satisfied, is given by
\begin{equation}
W_1^{(c)}=\frac{1}{\beta_H}\log\left(\frac{\mathcal{Z}_H }{e^{c\beta_H \mathcal{Z}_C}+e^{-\beta_H \mathcal{E}}} \right)+c
\end{equation}
where we have used the expressions in \eqref{cboundedwork qubit}. If inequality \eqref{carnot2} is satisfied for the second part of the cycle, then we get 
\begin{equation}
W_2^{(c)}=\frac{1}{\beta_C}\log\left(\frac{\mathcal{Z}_C }{e^{-c\beta_C \mathcal{Z}_H}+e^{-\beta_C \mathcal{E}}} \right)-c
\end{equation}
Using $Q_H = \Delta U + W^{(c)}_1$, the efficiency is given by 
\begin{equation}
\eta^{(c)}_1 = \frac{W^{(c)}_1+F(\rho_H,\beta_C)-F(\rho_C,\beta_C)}{\Delta U + W^{(c)}_1}
\end{equation}
if inequality \eqref{carnot1} is satisfied and 
\begin{equation}
\eta^{(c)}_2 =\frac{W^{(c)}_1+W^{(c)}_2}{\Delta U + W^{(c)}_1}
\end{equation}
if inequalities \eqref{carnot1} and \eqref{carnot2} are satisfied. As the temperature difference between hot and cold baths increases, the efficiency increases. In the unbounded case, the Carnot efficiency is bounded from above by 1. As $\beta_H/\beta_C\rightarrow 0$, inequality \eqref{carnot2} cannot be satisfied for any finite $c$. To find the upper bound of $\eta^{(c)}_1$ we make the change of variables $\beta_C/n = \beta_H$, i.e. $T_H = n T_C$ and $n$ is the ratio of the two bath temperatures. We then take $\beta_H \rightarrow \infty$ and $n\rightarrow \infty$. \\
\begin{equation}
\lim\limits_{n\rightarrow 0}\left( \lim\limits_{\beta_H\rightarrow \infty}\left( \eta^{(c)}_1 \right)\right)=1-\frac{\mathcal{E}}{2(2c + \mathcal{E} )}
\end{equation} 
Of course as $\mathcal{E}\rightarrow 0$ the efficiency tends to 1, but this is because we are no longer extracting any work (as the two Gibbs states are identical). As $c\rightarrow\infty$ we recover the Carnot efficiency upper bound, although for any finite $c$ we can never reach an efficiency of 1. Therefore all realistic engines of this form have an upper bound to their efficiency defined by their energy gap and fragility $c$. This represents a general physical upper bound on the efficiency of a the qubit engine, given by its ability to withstand fluctuations. Also notice that as $c\rightarrow 0$ the maximal efficiency is bounded from below by $1/2$. For $c=0$ the efficiency is zero, as no work can be extracted. Therefore we observe, in this model at least, a genuine discontinuity between the capabilities of a thermal machine in the single-shot regime. Allowing for arbitrarily small fluctuations in principle will allow the engine to reach a maximum efficiency that is bounded from below by $1/2$. But if we demand that the work is deterministic, the efficiency is zero.\\

\end{document}